\documentclass[a4paper,onecolumn,11pt,accepted=2026-07-20]{quantumarticle}
\pdfoutput=1
\usepackage[utf8]{inputenc}
\usepackage[english]{babel}

\usepackage{amsfonts}
\usepackage{amsmath}
\usepackage{physics}
\usepackage{bm}
\usepackage{graphicx}
\usepackage{latexsym}
\usepackage{hyperref}
\usepackage[capitalize]{cleveref}
\usepackage{float}
\usepackage[T1]{fontenc}
\usepackage{color,graphicx}
\usepackage{array}
\usepackage{enumerate}
\usepackage{amsmath}
\usepackage{amssymb}
\usepackage{amsthm}
\usepackage{bbm} 
\usepackage{pgfplots}
\usepackage{pgf}
\usepackage{graphicx}
\usepackage{tikz}
\usetikzlibrary{patterns}
\usetikzlibrary{arrows.meta}
\usepgfplotslibrary{patchplots} 
\usetikzlibrary{pgfplots.patchplots} 
\pgfplotsset{width=9cm,compat=1.5.1}
\usepackage{diagbox}

\definecolor{myblue}{RGB}{0,0,128}
\definecolor{myblue2}{RGB}{0,32,96}
\definecolor{myblue3}{RGB}{0,64,64}
\definecolor{myblue4}{RGB}{0,96,32}
\definecolor{myblue5}{RGB}{0,128,0}
\definecolor{myblue6}{RGB}{32,96,0}
\definecolor{myblue7}{RGB}{64,64,0}
\definecolor{myblue8}{RGB}{96,32,0}
\definecolor{myblue9}{RGB}{128,0,0}

\usepackage{mathtools}
\usepackage{amsthm}

\newtheorem{theorem}{Theorem}

\newtheorem{lemma}[theorem]{Lemma}
\newtheorem{corollary}[theorem]{Corollary}
\newtheorem{proposition}[theorem]{Proposition}

\newtheorem{conjecture}[theorem]{Conjecture}

\usepackage{xcolor}
\usepackage{makeidx}
\usepackage{caption}
\usepackage{subcaption}
\usepackage{scalerel}
\usepackage[normalem]{ulem}
\definecolor{ballblue}{rgb}{0.13, 0.67, 0.8}

\newcommand{\diag}{\operatorname{diag}}

\NewDocumentCommand{\newstuff}{ +m O{} } {{\color{blue}#1}}

\def\C{\mathbb{C}}
\def\M{\mathbb{M}}

\def\A{\mathcal{A}}

\def\B{\mathcal{B}}

\def\E{\mathcal{E}}

\def\S{\mathcal{S}}
\def\T{\mathcal{T}}

\def\W{\mathcal{W}}

\def\Y{\mathcal{Y}}

\begin{document}

\title{Quantum nonlocality without entanglement and state discrimination measures}

\author{Shayeef Murshid}
    \affiliation{Electronics and Communication Sciences Unit, Indian Statistical
    Institute, 203 B. T. Road, Kolkata 700108, India}
    \email{shayeef.murshid91@gmail.com}
\author{Tathagata Gupta}
    \affiliation{Physics and Applied Mathematics Unit, Indian Statistical
    Institute, 203 B. T. Road, Kolkata 700108, India}
    \email{tathagatagupta@gmail.com}
    \orcid{0009-0005-6319-1590}
    \thanks{Present address: Department of Physics, Indian Institute of Technology Madras, Chennai 600036, India}
\author{Vincent Russo}
    \affiliation{Unitary Foundation, 315 Montogomery St, Fl 10 San Francisco,
    California 94104, USA} \email{vincent@unitary.foundation}
    \orcid{0000-0002-8952-8981}
\author{Somshubhro Bandyopadhyay}
    \affiliation{Department of Physical Sciences, Bose Institute, EN 80,
    Bidhannagar, Kolkata 700091, India} 
    \email{som@jcbose.ac.in}
    \orcid{0000-0002-4678-6331}

\maketitle
\begin{abstract} 
   An ensemble of product states is said to exhibit “quantum nonlocality without entanglement” if it cannot be optimally discriminated using local operations and classical communication (LOCC). We show that this property can depend on the chosen discrimination measure. Specifically, we construct a family of ensembles, each consisting of six linearly independent, equally probable bipartite product states, for which LOCC fails to achieve optimal minimum-error discrimination but succeeds in achieving optimal unambiguous discrimination. We further extend our construction to multipartite systems and provide strong numerical evidence that a similar separation between local and global optima is present for minimum-error discrimination, but not for unambiguous discrimination.  
\end{abstract}

\section{Introduction}
\label{sec:introduction}

Composite quantum systems may exhibit nonlocal properties. For example, the
celebrated Bell nonlocality~\cite{bell1966problem,brunner2014bell} arises from
entangled states through violations of Bell-type inequalities. Somewhat less
known, though well-studied, is quantum nonlocality without
entanglement~\cite{bennett1999quantum}. This nonlocality, which may be viewed as dual
to the Bell type, manifests in state discrimination problems in the
distant-lab paradigm of quantum information theory.

Suppose two distant observers, Alice and Bob, share a state chosen from an
ensemble of product states 
\begin{equation} 
    \label{eq:ensemble}
    \E_{\psi} =\left\{
        \left(\eta_{i},\ket{\psi_i} = 
        \ket{a_i} \otimes\ket{b_i} \right):i=1, \ldots, N
    \right\},
\end{equation} 
where $\eta_{i}$ is the prior probability associated with $\ket{\psi_i}$. They
have complete knowledge of $\E_{\psi}$ but do not know which particular $\ket{\psi_i}$ they share. Their objective is to determine this ``unknown''
state as well as possible. 

The task of determining an unknown state chosen from a known set of states is
the standard state discrimination problem.  However, in this case, Alice and Bob
being physically separated cannot perform any measurement of their choice. The
only measurements they can implement belong to the class of \emph{local
operations and classical communication} (LOCC); that is, they are free to
perform quantum operations on their local systems and communicate via classical
channels but cannot exchange quantum systems. Thus, the question is: Can they
optimally discriminate the given states by LOCC? By optimal discrimination, we
mean attaining the global optimum corresponding to a measure of
state discrimination; in particular, if the given states are orthogonal, optimal
implies perfect discrimination~\cite{bergou200411, chefles2000quantum,
barnett2009quantum}. 

One might expect optimal discrimination of product states, parts of which
may have been prepared separately, would be possible by LOCC. However, it turns out
this is not always the case~\cite{peres1991optimal}. There exist product
ensembles, members of which cannot be optimally discriminated by LOCC, even when
they are mutually orthogonal~\cite{bennett1999quantum}. Such ensembles are said
to exhibit \emph{nonlocality without entanglement}. Here, nonlocality refers to
the fact that joint measurements can extract more information about the quantum
state than LOCC protocols.

One could express this nonlocality in terms of a well-defined state
discrimination measure, such as probability of success or fidelity~\cite{bandyopadhyay2013tight}. Let $f$ be such a measure. For an
ensemble $\E_{\psi}$ and a measurement $\M$ (a POVM), the quantity $0\leqslant
f\left(\E_{\psi};\M\right)\leqslant1$ then tells us how well the members of
$\E_{\psi}$ can be discriminated by $\M$.

Let $f_{_{G}}\left(\E_{\psi}\right)$ and
$f_{_{L}}\left(\E_{\psi}\right)$ be the respective global and local optimum, defined as
\begin{equation}
    \label{f(G)}
    f_{_{G}}\left(\E_{\psi}\right)
    =\sup_{\M}f\left(\E_{\psi};\M\right)
    \quad \text{and} \quad
    f_{_{L}}\left(\E_{\psi}\right) 
    =\sup_{\M\in\text{LOCC}}f\left(\E_{\psi};\M\right).
\end{equation}
Since LOCC is a strict subset of all measurements one may perform
on a composite system, we have 
\begin{equation}
    \label{eq:f(L)<=00003Df)G)}
    f_{_{L}}\left(\E_{\psi}\right) \leqslant
    f_{_{G}}\left(\E_{\psi}\right).
\end{equation}
Therefore, if $\E_{\psi}$ is nonlocal in the sense discussed
above, the above inequality must be strict
\begin{equation}
    \label{f(L)<f(G)}
    f_{_{L}}\left(\E_{\psi}\right) 
    <f_{_{G}}\left(\E_{\psi}\right).
\end{equation}
To check whether a given product ensemble is nonlocal or not, one therefore needs to
compute both optima explicitly. However, for orthogonal states,
computing the local optimum may not be necessary as long as one can show whether or not the states can be perfectly discriminated by LOCC.

Suppose the states $\ket{\psi_i} \in \E_{\psi}$ form an
orthonormal basis. Then perfect discrimination is always possible by a separable
measurement~\cite{bandyopadhyay2015limitations}, the measurement operators being
$\ketbra{\psi_i}{\psi_i} = \ketbra{a_i}{a_i} \otimes \ketbra{b_i}{b_i}$, $i = 1, \ldots, N$. Nevertheless, there exist orthogonal product bases,
elements of which cannot be perfectly discriminated by
LOCC~\cite{bennett1999quantum}. This shows that LOCC is a strict subset of
separable measurements (SEP), which, of course, is a strict subset of global
measurements~\cite{chitambar2014everything}. This can be conveniently expressed
as 
\begin{equation}
    \text{LOCC}\subset\text{SEP}\subset\text{GLOBAL}.
\end{equation}
Thus for any product ensemble $\E_{\psi}$ (or for any ensemble, for
that matter) it holds that
\begin{equation}
    \label{fL<=00003DfS<=00003DfG}
    f_{_{L}}\left(\E_{\psi}\right)\leqslant
    f_{_{S}}\left(\E_{\psi}\right)\leqslant
    f_{_{G}}\left(\E_{\psi}\right),
\end{equation}
where $f_{_{S}}\left(\E_{\psi}\right)=\sup_{\M\in\text{SEP}}f\left(\E_{\psi};\M\right)$.
If at least one of the above inequalities is strict, then $\E_{\psi}$
is nonlocal. For example, if $\E_{\psi}$ is an orthonormal product
basis and is nonlocal, it implies that 
\begin{equation}
    f_{_{L}}\left(\E_{\psi}\right)<f_{_{S}}\left(\E_{\psi}\right)=f_{_{G}}\left(\E_{\psi}\right)=1.
\end{equation}

\subsection*{Motivation}

In this paper we will consider measures corresponding to
minimum-error~\cite{holevo1973statistical} and unambiguous
discrimination~\cite{ivanovic1987differentiate,dieks1988overlap,peres1988differentiate}.
The former minimizes the average error and applies to any set of states; the
corresponding measure is the success probability, the maximum probability that
the unknown state is correctly determined~\cite{bae2013structure}. The latter
strategy, however, can only be applied to sets of states satisfying a certain
condition. If this condition is met, the unknown state can be determined,
without error, with a nonzero probability; for example, a set of pure states can
be unambiguously discriminated if and only if they are linearly
independent~\cite{chefles1998unambiguous}. Note that in the
minimum-error case, the conclusion could be erroneous, whereas in the latter,
there's no room for error -- a measurement outcome either correctly identifies
the state or is inconclusive. 

Once an ensemble is specified, one might expect the relationships between
different measurement optimas to be measure-agnostic. However, this turns out
not to be the case. Consider the double trine ensemble consisting of three
equiprobable two-qubit product states~\cite{peres1991optimal}:
\begin{equation}
    \label{ensemble-1}
    \E_{\alpha} = 
    \left\{
        \left(\frac{1}{3},\ket{\alpha_i} \otimes \ket{\alpha_i} \right)
        : i = 1, 2, 3 
    \right\},
\end{equation}
where
\begin{eqnarray}
    \label{eq:trine}
    \ket{\alpha_1} = \ket{0}, & 
    \ket{\alpha_2} = \displaystyle-\frac{1}{2}\ket{0} -\frac{\sqrt{3}}{2}\ket{1}, & 
    \ket{\alpha_3} =-\frac{1}{2}\ket{0} +\frac{\sqrt{3}}{2}\ket{1}.
\end{eqnarray}
Let $p^{me}\left(\E_{\alpha}\right)$ and $p^{ud}\left(\E_{\alpha}\right)$ denote
the (optimum) success probabilities for minimum-error and unambiguous
discrimination, respectively. Then it holds
that~\cite{chitambar2013revisiting,chitambar2013local}
\begin{align}
    p_{_{L}}^{me}\left(\E_{\alpha}\right) & <p_{_{S}}^{me}\left(\E_{\alpha}\right)=p_{_{G}}^{me}\left(\E_{\alpha}\right),\label{double-trine-minimum-error}\\
    p_{_{L}}^{ud}\left(\E_{\alpha}\right) & =p_{_{S}}^{ud}\left(\E_{\alpha}\right)<p_{_{G}}^{ud}\left(\E_{\alpha}\right).\label{double trine unambiguous}
\end{align}
Note that $\E_\alpha$ exhibits nonlocality without entanglement for both measures. However, the relationships between measurement optimas are different. In the minimum-error case, LOCC is suboptimal to SEP which achieves the global optimum; in the unambiguous case, LOCC and SEP are equally good but remain suboptimal to global measurements. This
led the authors of~\cite{chitambar2013local} to ask the following question:

\emph{Do product ensembles exist that exhibit nonlocality without entanglement
for one discrimination measure but not for another?}

In this paper, we provide a positive answer to this question. Specifically, we construct a
family of product ensembles, each consisting of six linearly
independent, equally probable bipartite product states, for which LOCC fails to achieve
optimal minimum-error discrimination, but achieves optimal unambiguous
discrimination. Thus, these ensembles exhibit nonlocality without entanglement
for minimum-error discrimination, but not for unambiguous discrimination. 

The main result is presented in Section \ref{sec:results}, with proofs provided in Sections \ref{sec:ME} and \ref{sec:UD}. We further extend our construction to multipartite systems. Our numerical investigations provide strong evidence that the resulting multipartite product ensembles exhibit a similar behavior: they remain nonlocal with respect to minimum-error discrimination, but not with respect to unambiguous discrimination. These are discussed in Section \ref{sec:general}. We conclude in Section \ref{sec:conclusion}.

\section{Results}
\label{sec:results}

Consider the following set of states
\begin{equation}
    \label{set S}
    \S =\left\{ \ket{\psi_i}
    \in\C^{3}:i=1,2,3,\;\left\langle \psi_{i}\vert\psi_{j}\right\rangle
    =s\in\left(0,1\right)\;\text{for }i\neq j\right\} .
\end{equation}
The above set of states is linearly independent. This follows from the result
that a set of $N\leqslant d$ pure states in $\C^{d}$ with pairwise real
and equal inner products is linearly independent if and only if the inner
product lies in the interval $\left(-\frac{1}{N-1},1\right)$ \cite{roa2011conclusive}. 

Note that $\S$ is a set of three symmetric pure states. While discrimination between two states is well understood with explicit formulas for the optimal success probability known for both minimum-error \cite{helstrom1969quantum,holevo1974remarks} and unambiguous discrimination \cite{jaeger1995optimal}, the case is quite different for three-state discrimination. Nevertheless, the three-state problem has been investigated under various constraints \cite{sun2001optimum, sugimoto2010complete} and under different settings. In particular, minimum-error and unambiguous discrimination of three optical coherent states have been studied in Refs. \cite{nakahira2018optimal, nakahira2019local} in both one-way LOCC and global measurement scenarios. Additionally, online strategies for the unambiguous identification of three symmetric pure states were considered in Ref. \cite{sentis2022online}.

Define the following set of product states
\begin{equation}
    \label{eq:six-state-set}
    \T = \left\{ 
    \ket{\psi_{i}} \otimes \ket{\psi_{j}} \in\C^{3} \otimes \C^{3}: 
    \ket{\psi_i}, \ket{\psi_j} \in \S, \;i\neq j,\;i,j=1,2,3
    \right\}.
\end{equation}
The set $\T$ is also linearly independent. That is because its elements can
be unambiguously discriminated, and hence, must be linearly independent
(a given set of pure states can be unambiguously discriminated if and only
if they are linearly independent) \cite{chefles1998unambiguous}. 

We study minimum-error and unambiguous discrimination of the members of the ensemble 
\begin{equation}
    \label{E(T)}
    \E_{\T} =\left\{
    \left(\frac{1}{6},\ket{\psi_i} \otimes \ket{\psi_j} \right):i\neq j,\;i,j=1,2,3\right\}.
\end{equation}
The main result is stated as follows.
\begin{theorem}
    \label{thm:main}
    Let $p^{me}\left(\E_{\T}\right)$ and
    $p^{ud}\left(\E_{\T}\right)$ denote the respective
    (optimum) success probabilities for minimum-error and unambiguous
    discrimination of the members of $\E_{\T}$. Then
    \begin{align}
        p_{_{L}}^{me}\left(\E_{\T}\right) & <p_{_{G}}^{me}\left(\E_{\T}\right),\label{thm-1-ME}\\
        p_{_{L}}^{ud}\left(\E_{\T}\right) & =p_{_{G}}^{ud}\left(\E_{\T}\right).\label{thm-1-UD}
    \end{align}
\end{theorem}
Thus $\E_{\T}$ is nonlocal with respect to minimum-error
discrimination but not unambiguous discrimination. 

We prove Eqs. \eqref{thm-1-ME} and \eqref{thm-1-UD} in Sections \ref{sec:ME} and \ref{sec:UD}, respectively. To prove \eqref{thm-1-ME}, we first reduce the problem of optimal LOCC discrimination of the ensemble $\E_\T$ to that of perfect LOCC discrimination of a corresponding set of orthogonal states, using a theorem from \cite{chitambar2013revisiting}. We then show that these orthogonal states cannot be perfectly discriminated by LOCC, which implies that the states in $\E_\T$ cannot be optimally discriminated using LOCC. Consequently, in the minimum-error setting, the optimal success probability achievable by LOCC is strictly smaller than the global optimum. 

To prove \eqref{thm-1-UD}, we explicitly compute the global optimum and show that it can be achieved using LOCC. 

\section{Minimum-error discrimination of $\E_\T$}\label{sec:ME}
Given an ensemble $\E_{\rho}=\left\{ \eta_{i},\rho_{i}\right\} _{i=1}^{N}$ and a
measurement $\M=\left\{ M_{1},\dots,M_{N}\right\} $, which is a collection of
positive operators forming a resolution of the identity, the error probability
is given by
\begin{align}
    p_{error}\left(\E_{\rho},\M\right) & =\sum_{\begin{array}{c}
    i,j=1\\
    i\neq j
    \end{array}}^{N}\eta_{i}\Tr\left(M_{j}\rho_{i}\right).\label{p-error}
\end{align}
In minimum-error discrimination, the goal is to find a measurement that minimizes $p_{error}\left(\E_{\rho},\M\right)$, thereby maximizing
the success probability. The minimum-error probability is given by
\begin{equation}
    \label{p-min-error}
    p_{error}\left(\E_{\rho}\right)
    =\min_{\M}p_{error}\left(\E_{\rho},\M\right)
\end{equation}
 and the success probability
\begin{align}
    \label{success-prob-min-error}
    p^{me}\left(\E_{\rho}\right) &
    =1-p_{error}\left(\E_{\rho}\right).
\end{align}
The advantage of this approach is that it applies to any set of
states. Finding the optimal
solution for an arbitrary ensemble, however, is
hard. Nevertheless, if the given states are pure and linearly independent, as in our case, an optimal measurement
consists of orthonormal, rank one projectors~\cite{mochon2006family}. This
result was subsequently strengthened by the following theorem.

\begin{theorem}[\cite{chitambar2013revisiting}]\label{thm:chitambar}
    Let $\E_{\chi}=\left\{ \eta_{i},\ket{\chi_{i}} \right\}_{i=1}^{N}$ be an
    ensemble of linearly independent pure states spanning a space $\Y$. There
    exists a unique orthonormal basis $\left\{\ket{\xi_{i}}
    \right\} _{i=1}^{N}$ of $\Y$ such that a measurement that achieves optimal
    minimum-error discrimination of the members of $\E_{\chi}$ also perfectly
    discriminates the states $\ket{\xi_{1}}, \ldots, \ket{\xi_{N}}$ and vice versa. 
\end{theorem}

The above theorem reduces the problem of optimal minimum-error discrimination of
linearly independent pure states to that of perfect discrimination of mutually
orthonormal pure states. We will make use of this fact to prove \eqref{thm-1-ME}. 

\subsection{Proof of LOCC suboptimality}

The following proposition follows from Theorem \ref{thm:chitambar}.
\begin{proposition}\label{prop}
    Let $\W$ be the subspace of $\C^{3}\otimes\C^{3}$
spanned by the elements of $\T$. There exists a unique orthonormal basis $\tilde{\B}$ of $\W$ such that a measurement that achieves the global optimum
$p_{_{G}}^{me}\left(\E_{\T}\right)$ perfectly discriminates the elements of $\tilde{\B}$ and vice versa.
\end{proposition}

\begin{corollary}\label{cor}
    LOCC achieves the global optimum $p_{_{G}}^{me}\left(\E_{\T}\right)$ if and only if the elements of $\tilde{\B}$ can be perfectly discriminated by LOCC.
\end{corollary} In what follows, we will find this unique orthonormal basis and show that its elements cannot be perfectly discriminated by LOCC. Therefore from Corollary \ref{cor}, optimal minimum-error discrimination of the elements of $\E_{\T}$ is not possible by LOCC. 

For ease of understanding, we will use the following notation to represent the
elements of $\T$:
\begin{eqnarray*}
    \ket{\phi_{1}} = \ket{\psi_{1}} \otimes \ket{\psi_{2}} & \;\; & 
    \ket{\phi_{2}} = \ket{\psi_{1}} \otimes \ket{\psi_{3}} \\
    \ket{\phi_{3}} = \ket{\psi_{2}} \otimes \ket{\psi_{1}}  & \;\; & 
    \ket{\phi_{4}} = \ket{\psi_{2}} \otimes \ket{\psi_{3}} \\
    \ket{\phi_{5}} = \ket{\psi_{3}} \otimes \ket{\psi_{1}} & \;\; & 
    \ket{\phi_{6}} =\ket{\psi_{3}} \otimes \ket{\psi_{2}} 
\end{eqnarray*}
and therefore our ensemble can be written as $\E_{\T}=\left\{ \frac{1}{6},\ket{\phi_{i}} \right\}_{i=1}^{6}$.

Since the states $\ket{\phi_{i}}$ are linearly independent, the optimal
measurement on $\W$ consists of orthogonal rank-one
projectors~\cite{mochon2006family}. Let this measurement be $\left\{
E_{i}=\ketbra{e_i}{e_i}\right\}_{i=1}^{6}$, where $\Tr\left(E_{i}E_{j}\right) =
\delta_{ij}$ for $i, j = 1, \ldots, 6$, and
$\sum_{i=1}^{6}E_{i}=\mathbf{1}_{\W}$. Thus the optimal measurement on
$\C^{3}\otimes\C^{3}$ is given by $\left\{
E_{1},\dots,E_{6},\left(\mathbf{1}_{3\times3}-\mathbf{1}_{\W}\right)\right\} $,
where $\left(\mathbf{1}_{3\times3}-\mathbf{1}_{\W}\right)$ is the projector onto
$\W^{\perp}$. It follows that $\tilde{\B}=\left\{ \ket{e_{1}}, \ldots,
\ket{e_{6}} \right\} $ must be the unique orthonormal basis of $\W$ mentioned in Prop. \ref{prop}. 

Our objective is to find the measurement $\left\{ E_{i} =
\ketbra{e_i}{e_i} \right\}_{i=1}^{6}$ as it would immediately lead to
$\tilde{\B}$. Fortunately, this measurement turns out to be the well-known
square-root measurement (SRM), also known as pretty-good
measurement~\cite{hausladen1994pretty}. The SRM operators are one-dimensional
projectors~\cite{sasaki1998quantum,eldar2001quantum,dalla2015optimality}
\begin{equation}
    \label{SRM operators}
    \mu_{i} =\ketbra{\mu_{i}}{\mu_{i}}, \quad i = 1, \ldots, 6
\end{equation}
satisfying $\ip{\mu_{i}}{\mu_{j}} =\delta_{i,j}$ for all $i,j=1,\dots,6$ and
$\sum_{i=1}^{6}\mu_{i}=\mathbf{1}_{\W}$, where
\begin{equation}
    \label{SRM vectors}
    \ket{\mu_{i}} = 
    \rho^{-1/2}\frac{1}{\sqrt{6}}\ket{\phi_{i}} \quad \text{ and
    }\quad 
    \rho=\frac{1}{6}\sum_{i=1}^{6}\ketbra{\phi_i}{\phi_i}.
\end{equation}
Note that the vectors $\left\{\ket{\mu_{i}} \right\}_{i=1}^{6}$ form an
orthonormal basis of $\W$. 
\begin{proposition}\label{thm:SRM}
    The square-root measurement is optimal for minimum-error discrimination of
    the members of $\E_\T = \left\{ \frac{1}{6},\ket{\phi_{i}} \right\}_{i=1}^{6}$. 
\end{proposition}
\begin{proof}
    We know that for minimum-error discrimination
    of a set of linearly independent pure states, the SRM is optimal when all the
    diagonal elements of the square root of the Gram matrix of the states are
    equal~\cite{sasaki1998quantum}.  A straightforward calculation shows this is indeed the case (proof in
    Appendix~\ref{gram_matrix}). Therefore our claim is
    proved.
\end{proof}
We therefore have $\tilde{\B}=\left\{ \ket{\mu_1}, \ldots, \ket{\mu_6}
\right\}$. Note that while $\ket{\phi_i}$ are product vectors, $\ket{\mu_i}$
may not be. We now prove that the
vectors $\ket{\mu_1}, \ldots, \ket{\mu_6}$ cannot be perfectly discriminated by
LOCC. 

First we prove the following property of $\W$. 
\begin{lemma}\label{lem:OPB}
    $\W$ does not admit an orthogonal product basis. 
\end{lemma}
\begin{proof}
    The proof is by contradiction. Assume that $\B_\mathrm{OPB}$ is an orthogonal product basis of $\W$. Let $\ket{\psi'_1}$ be a unit vector orthogonal to $\ket{\psi_2}$ and $\ket{\psi_3}$. Since the subspace orthogonal to the span of $\{\ket{\psi_2},\ket{\psi_3}\}$ is one-dimensional,
    this vector $\ket{\psi'_1}$ is unique. 

    Let us now assume that $\ket{\alpha}\otimes\ket{\beta}$ is a product state in the subspace orthogonal to that spanned by $\B_\mathrm{OPB}\cup
    \{\ket{\psi'_1}\otimes\ket{\psi'_1}\}$. Therefore, $\ket{\alpha}\otimes\ket{\beta}$
    is orthogonal to every member of the set
    \begin{equation}
       \mathcal{A} = \{\ket{\psi'_1}\otimes\ket{\psi'_1}, \ket{\psi_i}\otimes\ket{\psi_j}: i,j=1,2,3, \; i \neq j\}.
    \end{equation}
    This can happen if $\ket{\alpha}$ is orthogonal to Alice's state or
    $\ket{\beta}$ is orthogonal to Bob's state, or both. Since $\A$ contains
    seven states, this means there are at least four states where either
    $\ket{\alpha}$ is orthogonal to Alice's state (and $\ket{\beta}$ is
    orthogonal to the rest)  or $\ket{\beta}$ is orthogonal to Bob's state (and
    $\ket{\alpha}$ is orthogonal to the rest). 
    
    We first consider the case where $\ket{\alpha}$ is orthogonal to Alice's side in exactly four states of $\A$. This can have two possible subcases.
    \begin{enumerate}
    \item \textbf{Subcase 1.1:} $\ket{\alpha}$ is nonorthogonal to
    $\ket{\psi'_1}$. 
    
    Given that $\ket{\alpha}$ cannot be orthogonal to all states in
    $\{\ket{\psi_i}\}_{i=1}^3$, it will be orthogonal to exactly two of them. Let them be $\ket{\psi_1}$ and $\ket{\psi_2}$. This means $\ket{\beta}$ is orthogonal to $\{\ket{\psi'_1},
    \ket{\psi_1},\ket{\psi_2}\}$. We now show that the set $\{\ket{\psi'_1},
    \ket{\psi_1},\ket{\psi_2}\}$ is linearly independent and as a result
    $\ket{\beta}$ cannot be orthogonal to all of them. 

    Let $\ket{\psi'_1}=a\ket{\psi_1}+b\ket{\psi_2}+c\ket{\psi_3}$ and note that when $c\neq 0$, the set $\{\ket{\psi'_1}, \ket{\psi_1},\ket{\psi_2}\}$ is linearly independent. Therefore when $c=0$, 
    \begin{equation}
        \begin{aligned}
            & \braket{\psi_2}{\psi'_1}=as+b=0, \text{ and } \\
            & \braket{\psi_3}{\psi'_1}=as+bs=0.
        \end{aligned}
    \end{equation}
    Since $s \in (0.1)$, it follows that $b=0$. This implies $\ket{\psi'_1}$ and $\ket{\psi_1}$ are linearly dependent,
    which is a contradiction. Therefore, $c$ is necessarily nonzero and as a
    result $\{\ket{\psi'_1}, \ket{\psi_1},\ket{\psi_2}\}$ is linearly
    independent. Note that a similar argument shows that $\{\ket{\psi'_1}, \ket{\psi_1},\ket{\psi_3}\}$ is linearly
    independent. 
    
    \item \textbf{Subcase 1.2:} $\ket{\alpha}$ is orthogonal to
    $\ket{\psi'_1}$. 
    
    In this case, $\ket{\psi'_1}\otimes\ket{\psi'_1}$ can be regarded as an
    element of the four-state set whose states on Alice's side are orthogonal to
    $\ket{\alpha}$. Therefore $\ket{\alpha}$ has to be orthogonal to at least
    two states in $\{\ket{\psi_i}\}_{i=1}^3$. From the previous case it can be
    noted that each of the sets $\{\ket{\psi'_1}, \ket{\psi_1},\ket{\psi_2}\}$ and
    $\{\ket{\psi'_1}, \ket{\psi_1},\ket{\psi_3}\}$ is linearly independent. The
    remaining set $\{\ket{\psi'_1}, \ket{\psi_2},\ket{\psi_3}\}$ is linearly
    independent since $\ket{\psi'_1}$ is orthogonal to $\ket{\psi_2}$ and
    $\ket{\psi_3}$. 
    \end{enumerate}

    In both subcases, $\ket{\alpha}$ needs to be
    orthogonal to a linearly independent spanning set of $\C^3$, which is
    impossible. Therefore the case where $\ket{\alpha}$ is orthogonal to Alice's side in exactly four states of $\A$ is ruled out. The remaining cases where $\ket{\alpha}$ is orthogonal to Alice's side
    in five or more states are easy to discard, since these require $\ket{\alpha}$ to be orthogonal to Alice's side in at least four states of $\A$. By symmetry, we can argue that there cannot be four or more states in $\A$ where $\ket{\beta}$ is orthogonal to Bob's side. Therefore, we have shown that
    there exists no product state orthogonal to the subspace spanned by the
    states $\B_\mathrm{OPB}\cup \{\ket{\psi'_1}\otimes\ket{\psi'_1}\}$. Since
    $\ket{\psi'_1}\otimes\ket{\psi'_1}$ is orthogonal to the subspace $\W$, this
    means that $\B_\mathrm{OPB}\cup \{\ket{\psi'_1}\otimes\ket{\psi'_1}\}$ forms
    an unextendible product basis (UPB). But we know that in $\C^3 \otimes \C^3$
    any UPB contains exactly five     elements~\cite{divincenzo2003unextendible,bej2021unextendible} and therefore
    we reach a contradiction. Consequently, there is no orthogonal product basis
    of $\W$. This completes our proof.  
\end{proof}

Lemma \ref{lem:OPB} tells us that $\tilde{\B}$ cannot be an orthonormal product
basis, which means at least one of its elements $\ket{\mu_i}$ must be entangled.
We will now show that all elements of $\tilde{\B}$ are entangled and also have
the same Schmidt rank.
\begin{lemma}\label{lem:local_unitary}
    $\tilde{\B}=\left\{ \ket{\mu_1},\ldots, \ket{\mu_6} \right\} $ is an
    entangled orthonormal basis of $\W$. Moreover, the vectors $\ket{\mu_i}$
    have the same Schmidt rank. 
\end{lemma}
To prove this, we will show that the vectors $\ket{\mu_{i}}$ are connected by
local unitaries. We suppress the $1/\sqrt{6}$ factor for convenience.
\begin{proof}
    Let us write $\ket{\mu_i}$ as
    \[\ket{\mu_i}=\rho^{-1/2}\ket{\psi_{i_1}}\ket{\psi_{i_2}}, \text{ where }
    i_1 \neq i_2\] and let $U_{{i_1}{j_1}}:\C^3 \rightarrow \C^3$ be the unitary
    operator that satisfies \[U_{{i_1}{j_1}}\ket{\psi_{i_1}}=\ket{\psi_{j_1}},
    U_{{i_1}{j_1}}\ket{\psi_{j_1}}=\ket{\psi_{i_1}}\] and is identity on the
    rest. Note that, since
    \[\rho=\frac{1}{6}\sum_{i=1}^6 \dyad{\phi_i}\] we have
    \[\left(U_{{i_1}{j_1}}\otimes U_{{i_2}{j_2}}\right)^{-1} \rho
    \left(U_{{i_1}{j_1}}\otimes U_{{i_2}{j_2}}\right)=\rho.\] Thus
    $U_{{i_1}{j_1}}\otimes U_{{i_2}{j_2}}$ commutes with $\rho$ and hence
    $\rho^{-1/2}$. Therefore, 
    \begin{equation*}
        \begin{aligned}
            \left(U_{{i_1}{j_1}}\otimes U_{{i_2}{j_2}}\right) \ket{\mu_i} &=\left(U_{{i_1}{j_1}}\otimes U_{{i_2}{j_2}}\right)\rho^{-1/2}\ket{\psi_{i_1}}\ket{\psi_{i_2}}\\ &=\rho^{-1/2}\left(U_{{i_1}{j_1}}\otimes U_{{i_2}{j_2}}\right)\ket{\psi_{i_1}}\ket{\psi_{i_2}} \\ &=\rho^{-1/2}\ket{\psi_{j_1}}\ket{\psi_{j_2}}\\ &=\ket{\mu_j}.
        \end{aligned}
    \end{equation*} This completes the proof.
\end{proof}

We now come to the main result of this section. 
\begin{lemma}\label{lem:main}
    The elements of the basis $\tilde{\B}=\left\{ \ket{\mu_1}, \ldots,
    \ket{\mu_6} \right\} $ cannot be perfectly discriminated by LOCC. 
\end{lemma} The proof is given in Appendix~\ref{main_lemma_proof}, a sketch of which is presented below. 

We use a result by Chen \textit{et.
al.}~\cite{chen2003orthogonality} which provides a necessary condition for a set
of orthogonal pure states to be LOCC distinguishable. 
\begin{lemma}(\cite{chen2003orthogonality}) \label{lem:chen}
    If the states $\ket{\mu_1}, \ldots, \ket{\mu_6}$
    can be perfectly discriminated by LOCC, then for every $i \in
    \{1,\dots,6\}$ it holds that \begin{equation}\label{eq:chen-sum}
        \ket{\mu_i}=\sum_j
    \ket{\alpha_{ij}}\otimes\ket{\beta_{ij}},\end{equation} where 
    \begin{subequations}\label{eq:chen-conditions}
    \begin{gather} \left(\bra{\alpha_{im}}\otimes\bra{\beta_{im}}\right)\ket{\mu_i} \neq 0 \text{ for all i}, \text{
        and,} \\ \left(\bra{\alpha_{jm}}\otimes\bra{\beta_{jm}}\right)\ket{\mu_i}=0 \; \text{ for all m
        and } \; j \neq i.
     \end{gather} \end{subequations}
\end{lemma} 
We know that $
\ket{\mu_1}, \ldots, \ket{\mu_6}$ are entangled. Now if they are LOCC distinguishable, there must be at least two
linearly independent product vectors appearing in the decomposition of each $\ket{\mu_i}$ such that these two conditions are satisfied. In Appendix \ref{main_lemma_proof} we show that this is impossible. 
   
As noted earlier, by virtue of Theorem~\ref{thm:chitambar}, a
measurement that achieves optimal minimum-error discrimination of the members of
$\E_{\T}$ must perfectly discriminate the elements of the
basis $\tilde{\B}$. Since LOCC cannot perfectly discriminate the
elements of the basis $\tilde{\B}$, by Corollary \ref{cor}  we have the following theorem: 
\begin{theorem}
    The optimal minimum-error discrimination of the elements of $\E_{\T}$ cannot be achieved by LOCC.  
\end{theorem} Therefore,
    $p_{_{L}}^{me}\left(\E_{\T}\right)<p_{_{G}}^{me}\left(\E_{\T}\right)$,
    proving \eqref{thm-1-ME} of Theorem~\ref{thm:main}.

\section{Unambiguous discrimination of $\E_\T$}\label{sec:UD}

Unambiguous discrimination
(UD)~\cite{ivanovic1987differentiate,dieks1988overlap,peres1988differentiate} of the members of
an ensemble $\E_\chi=\{\eta_i,\ket{\chi_i}\}_{i=1}^N$, involves constructing a measurement $E=\{E_i\}_{i=0}^N$ with $N+1$ outcomes such that 
\begin{equation}
    \label{eq:ud-condition}
    \Tr\left(E_i\dyad{\chi_j}\right)=p_i \delta_{ij}
\end{equation} 
for all $i,j \in \{1,\dots,N\}$ and $E_0$ corresponds to the
inconclusive outcome. Here $p_i$, called the
\textit{efficiency} for $\ket{\chi_i}$, denotes the probability that upon
receiving the input state $\ket{\chi_i}$ the measurement successfully identifies
it. When the efficiencies are demanded to be equal, i.e., $p_i=p$ for all $i
=1,\dots,N$, the corresponding task is referred to as \textit{equiprobable
unambiguous discrimination}~\cite{horoshko2019equiprobable}. Not all sets of pure states allow for unambiguous discrimination; a necessary and
sufficient condition for unambiguous discrimination of a set of pure states is that
the states should be linearly independent~\cite{chefles1998unambiguous}. 

The objective is to find a measurement that minimizes the probability of the inconclusive outcome
\begin{equation}
    p_?(\E_\chi,E) = \sum_{i = 1}^N \eta_i \Tr\left(E_0\dyad{\chi_i}\right). 
\end{equation} 
Thus, the maximum probability of conclusively identifying the state is given by
\begin{equation}
  p^{ud}(\E_\chi)= 1-\min_E p_?(\E_\chi,E)=\max_{\{p_i\}}\sum_{i=1}^N \eta_ip_i.
\end{equation}
Finding the optimal probability of success is hard for a general problem, with solutions being known only for the two-state case~\cite{jaeger1995optimal} and in some cases involving symmetries and
constraints~\cite{sugimoto2010complete,sun2001optimum}. However, it can be cast as a semidefinite program~\cite{boyd2004convex,eldar2003semidefinite} 
\begin{equation}
    \label{primal_sdp}
    \begin{aligned}
        &{\text{maximize}}
        & & \sum_{i=1}^N\eta_ip_i \\
        & \text{subject to}
        & & \Gamma -P \succeq 0  \\
        &&& P \succeq 0
    \end{aligned}
\end{equation} 
where $P = \diag(p_1, \dots, p_N)$ and $\Gamma$ is the Gram matrix of
the set of states having the entries $\Gamma_{ij}=\braket{\chi_i}{\chi_j}$. The
dual to the above problem is~\cite{gupta2024unambiguous}
\begin{equation}\label{dual_sdp}
    \begin{aligned}
        & \underset{Z,\vec{z}}{\text{minimize}}
        & & \Tr(\Gamma Z) \\
        & \text{subject to}
        & & z_i+\eta_i-Z_{ii}=0 \\
        &&& Z,\vec{z} \succeq 0.
    \end{aligned}
\end{equation}
where $Z_{ii}$ denotes the $i^{th}$ diagonal entry of $Z$. 

\subsection{Proof of LOCC optimality}

The states of $\T$ admit unambiguous discrimination since they
are linearly independent. We can lower bound the optimum probability of success
by using the following result which provides the optimum probability for
unambiguous discrimination of any set of equiprobable linearly independent
states whose pairwise inner products are equal and real. 
\begin{lemma}(\cite{roa2011conclusive})
    \label{lem:roa-opt-prob}
    Let $S_{N}=\left\{ \ket{\psi_i}: 2\leqslant i\leqslant
    N\right\} $ be a set of equally likely, linearly independent pure states
    with the property $\left\langle \psi_{i}\vert\psi_{j}\right\rangle =s$ for
    $i\neq j$, where $s\in\left(-\frac{1}{N-1},1\right)$. Then the optimum
    probability for unambiguous discrimination is
    \begin{alignat}{1}
        p & =
        \begin{cases}
            1-s, & s\in\left[0,1\right) \\
            1+\left(N-1\right)s, & s\in\left(-\frac{1}{N-1},0\right].
        \end{cases}
    \end{alignat}
\end{lemma} 
Using the above result, we can find a local protocol that unambiguously
discriminates the members of $\E_\T$ with probability $(1-s)^2$. Given a
state from $\E_\T$, one can successfully identify the state of the first component with probability $1-s$. This leaves the second system in one of the two
remaining states of $\S$ with equal probability, which again can be identified
with probability $(1-s)$. To conclusively identify the given state from $\E_\T$,
one has to obtain conclusive outcomes for both measurements, and hence, the
probability of success is $(1-s)^2$. Note that this is a local protocol that
uses only classical outcome of the first measurement to design the second measurement. Therefore, 
\begin{equation}
    p^{ud}_G(\E_\T) \geq (1-s)^2.
\end{equation} 
We prove \eqref{thm-1-UD} by showing that this is the maximum probability of unambiguous discrimination allowed by quantum theory.
\begin{lemma}\label{lem:T_eq_prob}
   The optimal success probability of unambiguous discrimination of the elements of $\E_\T$ is given by $(1-s)^2$.
\end{lemma}

\begin{proof}
   We will consider the optimization problem in Equation~\eqref{primal_sdp} for
   different ensembles of states that are generated by permutations of the set
   $\{1, \ldots, 6\}$. First note that if $\sigma$ is a permutation of
   \(\{1, 2, \ldots, 6\}\), then the set of states 
   
   \begin{equation}
        \T_\sigma = \left\{ \ket{\psi_{\sigma(i)}} \ket{\psi_{\sigma(j)}} \mid i,j \in \{1,2,3\} \text{ and } i \neq j \right\} 
    \end{equation} 
    has the same Gram matrix as that of $\T$. Therefore, we can denote the Gram
    matrix of $\T_\sigma$ by $\Gamma$ for any permutation $\sigma$. Moreover, if $P
    = \diag(p_1, \ldots, p_6)$ is the optimal solution for $\E_\T$, then the
    optimal solution for $\E_{\T_{\sigma}}$ is given by
    $P_\sigma=\diag\left(p_{\sigma(1)}, \ldots,
    p_{\sigma\left(6\right)}\right)$. 
    
    These solutions satisfy $\Gamma - P_\sigma \succeq 0$ for all permutations
    $\sigma$. These are $6!$ conditions, one for each permutation of the set
    $\{1, \ldots, 6\}$. We add them and use the fact that the sum of two
    positive semidefinite matrices is again positive semidefinite, to get
    \begin{equation}
        \Gamma - P_\mathrm{avg} \succeq 0
    \end{equation} 
    where 
    \begin{equation}
        P_\mathrm{avg} = \frac{1}{6!} \sum_\sigma P_\sigma = \lambda \mathbf{1}
    \end{equation}
    and $\lambda=\left(\sum_{i=1}^6 p_i\right)/6$. This shows that optimal
    unambiguous discrimination of the members of $\E_\T$ is achieved by an equiprobable
    unambiguous discrimination~\cite{horoshko2019equiprobable}; therefore, the optimization problem can be expressed as
    
    \begin{equation}
        \begin{aligned}
            &{\text{maximize}}
            & & \lambda \\
            & \text{subject to}
            & & \Gamma -\lambda \mathbf{1} \succeq 0  \\
            &&& \lambda \succeq 0.
        \end{aligned} 
    \end{equation} 
    
    This is a standard SDP whose solution is given by the minimum eigenvalue of
    $\Gamma$~\cite{boyd2004convex} that corresponds to the optimum success probability. It is straightforward to compute the eigenvalues
    of $\Gamma$ (this matrix is presented in Appendix~\ref{gram_matrix}). The
    eigenvalues are
    
    \begin{equation}
        \left\{1-s,1-s,(1-s)^2,1+s-2s^2,1+s-2s^2,1+2s+3s^2\right\}.
    \end{equation}
    
    It is easy to see that $(1-s)^2<1-s$. For the remaining two eigenvalues, observe that \begin{align*}
        1+s-2s^2 &= (1-s)^2+3s(1-s)>(1-s)^2, \text{ and } \\ 1+2s+3s^2&=(1-s)^2+2s(2+s)>(1-s)^2.
    \end{align*} Thus $(1-s)^2$ is the
    minimum eigenvalue of $\Gamma$, the Gram matrix of $\T$, over the entire interval
    $(0,1)$.
\end{proof} 

Since we have already established that there is a local protocol that succeeds in unambiguously discriminating the members of $\E_\T$ with probability $(1-s)^2$, we have the following theorem
\begin{theorem}
    Optimal unambiguous discrimination of the members of $\E_\T$ is achievable by LOCC. 
\end{theorem} It follows that $p_L^{ud}(\E_\T)=p_G^{ud}(\E_\T)$, proving \eqref{thm-1-UD}. This completes the proof of Theorem \ref{thm:main}.

\section{Extension to multipartite systems}\label{sec:general}

One may consider generalizing our construction. Perhaps the simplest
way to do so is by introducing a third party, Charlie, who is assigned
the remaining state from the set $\S$ as defined in Eq. \eqref{set S}, with
all three parties kept spatially separated. The resulting ensemble
is given by
\begin{align}
\mathcal{E} & =\left\{ \left(\frac{1}{6},|\psi_{i}\rangle\otimes|\psi_{j}\rangle\otimes|\psi_{k}\rangle\right):|\psi_{i}\rangle,|\psi_{j}\rangle,|\psi_{k}\rangle\in \S,;i\neq j\neq k,;i,j,k\in{1,2,3}\right\} .
\end{align}

Now the question is whether the elements of this ensemble can be optimally
distinguished by LOCC in the unambiguous discrimination setting, but
not in the minimum-error setting. For our subsequent discussions,
it would be helpful to work with a more generalized version -- one
that extends to $k$ spatially separated parties. 

Consider a set of $N$ linearly independent states in $\mathbb{C}^{N}$
with equal and real pairwise overlaps,
\begin{align} \label{set S_N}
\S_{N} & =\left\{ |\psi_{i}\rangle\in\mathbb{C}^{N}:i=1,2,\dots,N;\langle\psi_{i}|\psi_{j}\rangle=s\in(0,1)\ \text{for }i\neq j\right\} .
\end{align}

From this set, choose $k$ distinct elements to form the set of product
states,
\begin{align}
\mathcal{T}_{N,k} & =\left\{ |\psi_{i_{1}}\rangle\otimes|\psi_{i_{2}}\rangle\otimes\cdots\otimes|\psi_{i_{k}}\rangle:i_{1},i_{2},\dots,i_{k}\in\{1,\dots,N\};i_{a}\neq i_{b}\ \forall a\neq b\right\} .
\end{align}

The cardinality of this set is $|\T_{N,k}|=\frac{N!}{(N-k)!}={}^{N}P_{k}$,
which corresponds to the number of injective maps from $k$ positions
into $N$ labels. Distributing each subsystem to each of $k$ spatially
separated parties yields the ensemble
\begin{align} \label{eq:gensemble}
\mathcal{E}_{N,k} & =\left\{ \left(\frac{1}{^{N}P_{k}},|\psi_{i_{1}}\rangle\otimes\cdots\otimes|\psi_{i_{k}}\rangle\right):|\psi_{i_{j}}\rangle\in \S_{N}\right\} .
\end{align} 

One may note that the ensemble considered in the previous sections
corresponds to the case $\mathcal{E}_{3,2}$.

Thus the general problem is to investigate LOCC discrimination of
the elements of $\mathcal{E}_{N,k}$ for $N\geq3$ and $2\leq k\leq N$
for both minimum-error and unambiguous paradigms. Our numerical results
suggest that while LOCC achieves optimal unambiguous discrimination,
it remains sub-optimal for minimum-error discrimination. Notably,
in the minimum-error case, it appears that the global optimum is achieved
if and only if all parties come together --- in other words, LOCC
is sub-optimal across all bipartitions. 

\subsection{Minimum-error}
Recall that our proof of Theorem \ref{thm:main} in the minimum-error setting began
with an application of Theorem \ref{thm:chitambar}, which guaranteed the existence of
a unique orthonormal basis for the subspace spanned by the states
of $\mathcal{T}$, such that perfect local distinguishability of this
basis is equivalent to optimal LOCC discrimination of $\mathcal{E}_{\mathcal{T}}$.
We then showed that this basis coincides with the square-root measurement
(SRM) operators for $\mathcal{E_{T}}$, by establishing their optimality
through a property of the Gram matrix of $\mathcal{E_{T}}$ (Proposition
\ref{thm:SRM}). Next, we demonstrated that this basis is entangled. This allowed
us to invoke Lemma \ref{lem:chen} to conclude that the basis is not perfectly distinguishable
by LOCC, which in turn implies that LOCC cannot achieve the global
optimum for distinguishing the elements of $\mathcal{E_{T}}$.

In this subsection, we present evidence suggesting that all these
essential features of the ensemble $\mathcal{E_{T}}$ (equivalently,
$\mathcal{E}_{3,2}$) are present for the more general family $\mathcal{E}_{N,k}$,
for all $N\geq3$ and $2\leq k\leq N$.

We begin by observing that for two states 
$|\varphi\rangle = |\psi_{i_1}\rangle\otimes\cdots\otimes|\psi_{i_k}\rangle$ 
and 
$|\varphi'\rangle = |\psi_{j_1}\rangle\otimes\cdots\otimes|\psi_{j_k}\rangle$ 
in $\mathcal{T}_{N,k}$, the Gram matrix entry factorizes as
\begin{equation}
    \langle\varphi|\varphi'\rangle 
    = \prod_{\ell=1}^{k} \langle\psi_{i_\ell}|\psi_{j_\ell}\rangle
    = s^{\,|\{\ell : i_\ell \neq j_\ell\}|},
\end{equation}
since each factor equals $1$ if $i_\ell = j_\ell$ and $s$ otherwise. 
Consequently, the entry $G_{\varphi\varphi'}$ takes the value $s^m$, where 
$m \in \{0, 1, \ldots, k\}$ counts the number of tensor positions at which 
the two index sequences disagree. 

Now consider the symmetric group $\mathrm{Sym}(N)$ which acts on $\mathcal{T}_{N,k}$ 
by permuting the labels of the component states. For $\sigma \in \mathrm{Sym}(N)$,
\begin{equation}
    \sigma \cdot \bigl(|\psi_{i_1}\rangle\otimes\cdots\otimes|\psi_{i_k}\rangle\bigr)
    = |\psi_{\sigma(i_1)}\rangle\otimes\cdots\otimes|\psi_{\sigma(i_k)}\rangle.
\end{equation}
This action preserves every inner product between states in $\mathcal{T}_{N,k}$ 
and hence leaves the Gram matrix invariant under conjugation by the 
corresponding permutation matrices $P_\sigma$:
\begin{equation}
    P_\sigma\, G\, P_\sigma^T = G \quad \forall\; \sigma \in \mathrm{Sym}(N).
    \label{eq:general-invariance}
\end{equation}
From the uniqueness of the positive semidefinite square root of $G$ we have,
\begin{equation}
    P_\sigma\, G^{1/2}\, P_\sigma^T = G^{1/2} \quad 
    \forall\; \sigma \in \mathrm{Sym}(N).
\end{equation}
Additionally, the action of $\mathrm{Sym}(N)$ on $\mathcal{T}_{N,k}$ is 
\emph{transitive}: given any two states 
$|\psi_{i_1}\rangle\otimes\cdots\otimes|\psi_{i_k}\rangle$ and 
$|\psi_{j_1}\rangle\otimes\cdots\otimes|\psi_{j_k}\rangle$ in 
$\mathcal{T}_{N,k}$, there always 
exists a permutation $\sigma \in \mathrm{Sym}(N)$ mapping one index 
sequence to the other (since $\mathrm{Sym}(N)$ acts transitively on 
ordered $k$-tuples of distinct elements from $\{1,\ldots,N\}$). This means that for the diagonal entries we have
\begin{equation}
    \bigl(G^{1/2}\bigr)_{ii} 
    = \bigl(P_\sigma G^{1/2} P_\sigma^T\bigr)_{ii}
    = \bigl(G^{1/2}\bigr)_{\sigma(i),\sigma(i)}.
\end{equation}
Since this holds for arbitrary $\sigma$ and $i$, all diagonal 
entries of $G^{1/2}$ must be equal
\begin{equation}
    \bigl(G^{1/2}\bigr)_{ii} = c_{N,k}(s) \quad \forall\; i = 1, \ldots, \tfrac{N!}{(N-k)!},
    \label{eq:general-equal-diag}
\end{equation}
where $c_{N,k}(s) > 0$ is a constant depending on $N$, $k$, and the 
inner product $s$.

The above analysis implies that the SRM is optimal for discriminating $\E_{N,k}$. Since the states of $\T_{N,k}$ are linearly independent, the SRM operators form an orthonormal basis of the subspace in $(\C^N)^{\otimes k}$ that is spanned by the elements of $\T_{N,k}$, implying that Theorem \ref{thm:chitambar} is again applicable. To establish features analogous to the $\E_{3,2}$ case, it remains to show that these operators are entangled and cannot be perfectly distinguished by LOCC. For this, we resort to numerical methods. We
compare the success probability of the SRM with that achievable by
PPT measurements (i.e., measurements whose operators are positive under
partial transpose) for distinguishing $\E_{3,3}$, $\E_{4,2}$, and
$\E_{4,3}$. Since PPT measurements strictly contain LOCC and their
optimal performance can be computed via semidefinite programming
\cite{cosentino2013positive}, any gap immediately certifies LOCC
suboptimality.

The results, presented in Figure~\ref{fig:me_gap}, clearly reveal a gap in all cases. The SDPs were solved using CVXPY with the MOSEK solver, and the observed gaps---reaching magnitudes on the order of $10^{-2}$---persist across a dense grid of $30$ overlap parameter values, ruling out numerical artifacts. Moreover, since PPT measurements strictly contain LOCC, the gap between SRM and LOCC performance is at least as large.
\begin{figure}[h]
    \centering
    \includegraphics[width=0.85\columnwidth]{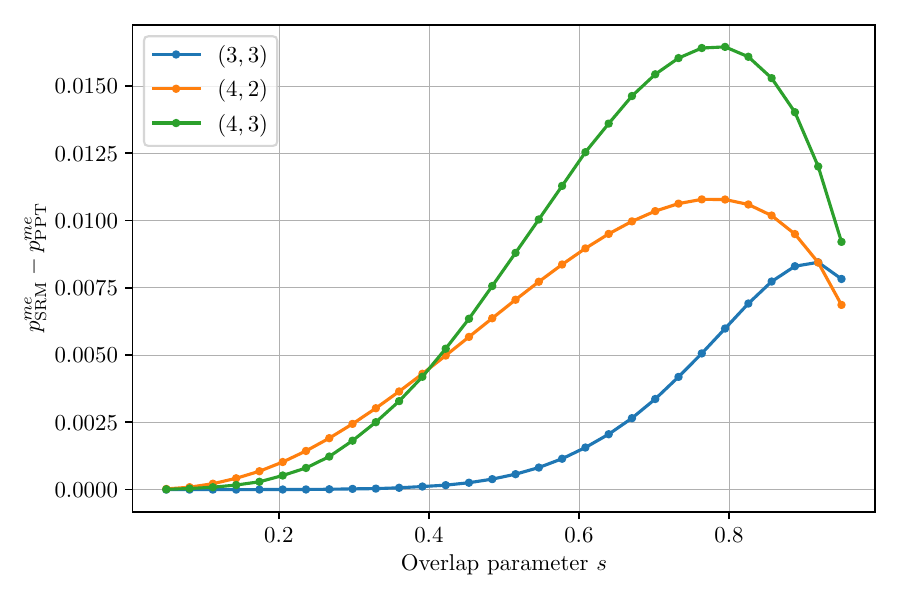}
    \caption{Gap between SRM and PPT success probabilities for minimum-error discrimination of $\E_{N,k}$ as a function of the overlap parameter $s$. A nonzero gap implies that PPT (and hence LOCC) measurements are suboptimal.}
    \label{fig:me_gap}
\end{figure}

Two remarks are in order. First, since the SRM outperforms PPT measurements, at least one SRM element must be entangled. By symmetry, there is no preferred element, suggesting that all SRM operators are entangled, as in the $\E_{3,2}$ case. Second, in solving the SDPs for PPT measurements, the bipartition is taken with respect to the first party and the rest of the system. However, due to the symmetry of the problem, this choice is equivalent to any other bipartition where one party is separated from the rest. This implies that even if all but one of the parties come together and perform LOCC, they cannot optimally distinguish the ensembles; in other words, LOCC is suboptimal across all bipartitions. This indicates that these ensembles exhibit genuine nonlocality without entanglement \cite{halder2019strong,xiong2023distinguishability} in the minimum-error setting.

\subsection{Unambiguous}The optimal unambiguous discrimination probability of $\E_{3,2}$ was derived in Lemma \ref{lem:T_eq_prob}. Its proof relied fundamentally on the permutation symmetry of the system, which dictated that the optimal discrimination is achieved through an equiprobable discrimination of the ensemble. By the arguments of the previous subsection, we observe that this symmetry is also present for all $\E_{N,k}$, suggesting that the reasoning of Lemma \ref{lem:T_eq_prob} extends to the general case. Under this symmetry assumption, the optimal success probability is determined by the minimum eigenvalue of the associated Gram matrix. In the following, we evaluate the eigenvalues for the ensembles $\E_{3,3}$, $\E_{4,2}$ and $\E_{4,3}$ and demonstrate that there exist LOCC protocols attaining the success probabilities predicted by these minimum eigenvalues.
\paragraph{$\E_{3,3}$:}The eigenvalues of the Gram matrix are plotted in Figure~\ref{fig:gram_33}. 
\begin{figure}[ht]
    \centering
    \includegraphics[width=0.85\columnwidth]{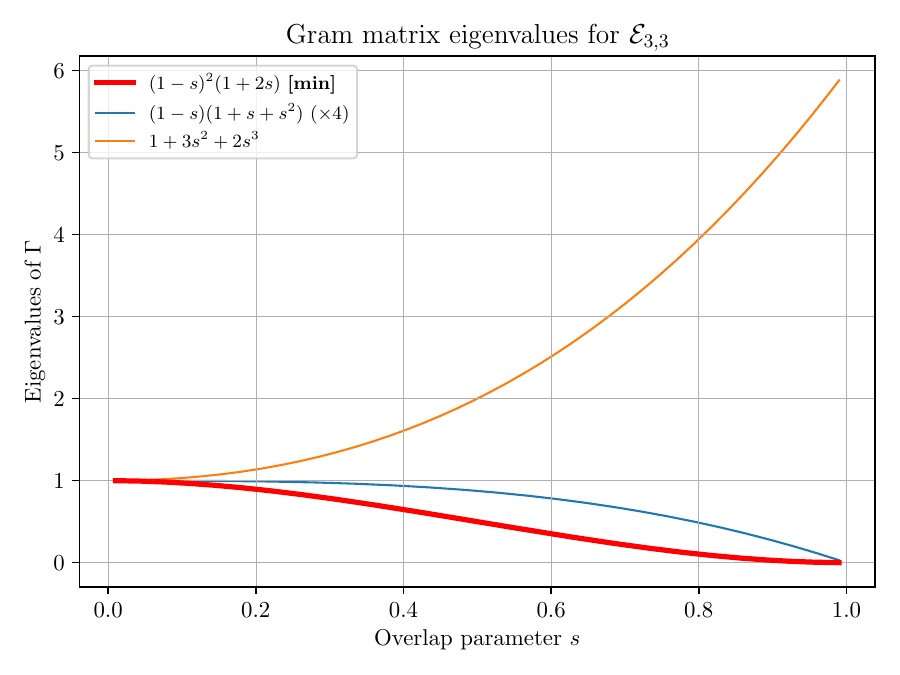}
    \caption{Eigenvalues of the Gram matrix $\Gamma$ for $\E_{3,3}$ as a function of the overlap parameter $s$. The lowest curve (red) corresponds to the minimum eigenvalue $(1-s)^2(1+2s)$, which equals the optimal success probability for unambiguous discrimination.}
    \label{fig:gram_33}
\end{figure}
The minimum eigenvalue is $(1-s)^2(1+2s)$, which determines the optimal success probability for unambiguous discrimination. An LOCC protocol achieving this success probability can be constructed in the following way. Begin by observing that $\E_{3,3}$ is obtained by appending the remaining state from $\S$ to the states of $\E_{3,2}$, and we already know that the states of $\E_{3,2}$ can be locally distinguished with a probability of $(1-s)^2$. Since identifying the states of the first two subsystems uniquely identifies the last, the success probability of discriminating $\E_{3,3}$ is at least $(1-s)^2$. But since the third component is available, and identifying any two of the three is enough to identify the joint state, the success probability for $\E_{3,3}$ increases. Even if one of the first two parties fails to identify their state, there is an overall success as long as the rest are successful. The probability with which exactly one of the first two parties fail, but the rest succeed is $s(1-s)^2$. This yields a total success probability of $(1-s)^2(1+2s)$, matching the minimum eigenvalue.

Note that by a similar analysis, ensembles of the form $\E_{N,N}$ can be locally distinguished with a success probability of $(1-s)^{N-1}(1+(N-1)s)$.

\paragraph{$\E_{4,2}$ and $\E_{4,3}$:}These cases are similar. The eigenvalues of the respective Gram matrices are plotted in Figure~\ref{fig:gram_42_43}. 
\begin{figure}[h]
    \centering
    \includegraphics[width=\textwidth]{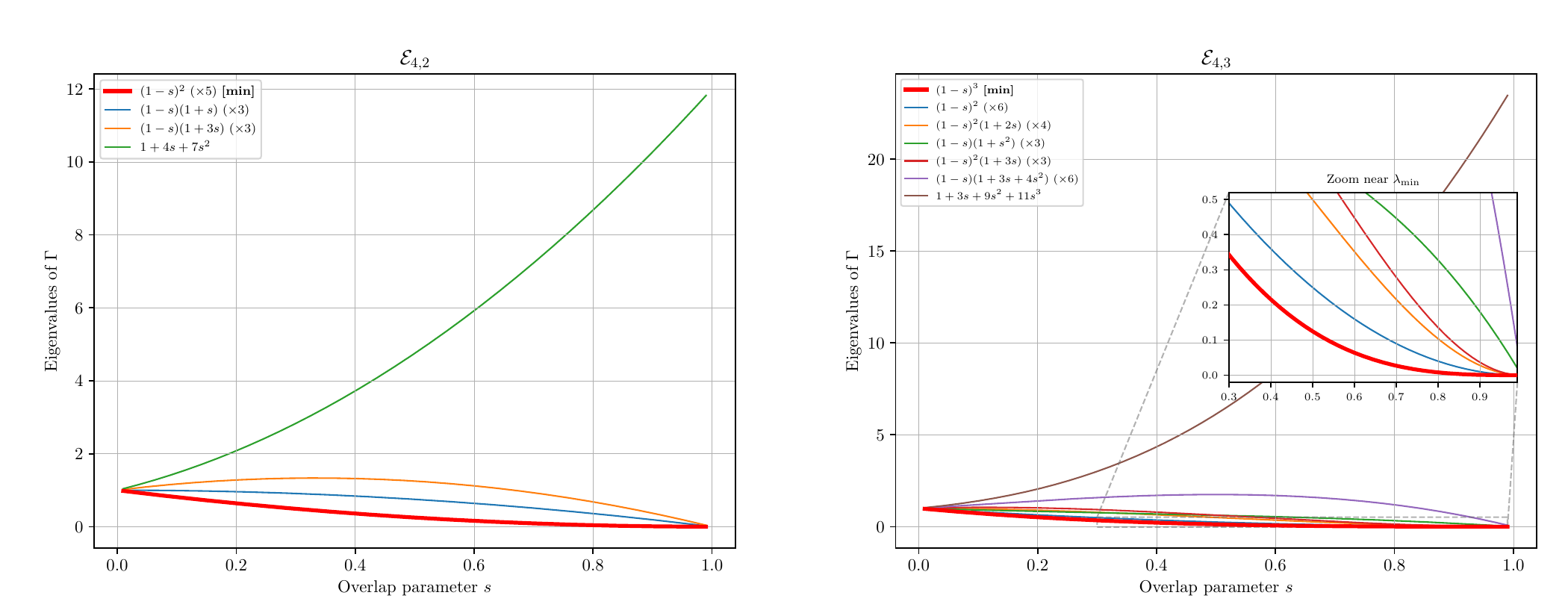}
    \caption{Eigenvalues of the Gram matrix $\Gamma$ for $\E_{4,2}$ (left) and $\E_{4,3}$ (right) as a function of the overlap parameter $s$. The minimum eigenvalue is highlighted in red in each panel. The inset in the right panel zooms near the minimum to resolve closely spaced branches.}
    \label{fig:gram_42_43}
\end{figure}
The minimum eigenvalues are $(1-s)^2$ for $\E_{4,2}$ and $(1-s)^3$ for $\E_{4,3}$. The inset in the $\E_{4,3}$ panel zooms near the minimum eigenvalue to clearly distinguish it from nearby branches. LOCC protocols achieving these success probabilities are similar to the one used in the analysis of $\E_{3,2}$. To successfully identify the given joint state, each party must be successful in identifying the state of their subsystem which, by Lemma \ref{lem:roa-opt-prob}, can be achieved with a probability $(1-s)$. This yields an overall success probability of $(1-s)^2$ and $(1-s)^3$ for the ensembles $\E_{4,2}$ and $\E_{4,3}$ respectively.

\subsection{Conjecture for multipartite systems} 
The evidence and discussion presented above motivates us to put forward the following conjecture.
\begin{conjecture}\label{conjecture}
    For the family of ensembles of multipartite pure product states $\E_{N,k}$ given by Equation \eqref{eq:gensemble} the following holds for all $N \geq 3$ and $2 \leq k \leq N$
    \begin{align}
        p_{_{L}}^{me}\left(\E_{N,k}\right) & <p_{_{G}}^{me}\left(\E_{N,k}\right),\\
        p_{_{L}}^{ud}\left(\E_{N,k}\right) & =p_{_{G}}^{ud}\left(\E_{N,k}\right).
    \end{align}
Moreover, the optimal unambiguous discrimination probability of $\E_{N,k}$ is given by 
\begin{equation}
p_{_{G}}^{ud}\left(\E_{N,k}\right) =
        \begin{cases}
            (1-s)^k, & \text{ when } k < N \\
            (1-s)^{N-1}(1+\left(N-1\right)s), & \text{ when } k=N.
        \end{cases}
    \end{equation}
\end{conjecture}

\section{Discussion}
\label{sec:conclusion}
We have shown that quantum nonlocality without entanglement – a property
ascribed to a set of product states that cannot be optimally discriminated by
LOCC – is not always a property of the underlying states but a consequence of
the chosen measure of state discrimination, i.e. a set of product states could be
nonlocal under one measure but not for another. Specifically, we presented a family
of sets of linearly independent product states, where each set has the following
property: LOCC is sub-optimal for minimum-error discrimination of its members
but optimal for unambiguous discrimination. The former is
proved by showing that LOCC does not satisfy a condition that a measurement
achieving the global optimum must; the latter is demonstrated by explicitly
computing the global optimum and then showing that an LOCC protocol attains
the same.

It is worth noting a curious feature of the ensemble $\mathcal{E}_{3,2}$
in the context of unambiguous discrimination. Although we described
the local protocol achieving the optimum success probability as an
LOCC scheme, the same success probability can in fact be achieved
without any classical communication between the parties. To see this,
suppose Alice performs an optimal measurement for $\S$ on her subsystem
and identifies the state as $|\psi_{1}\rangle$. If she communicates
this outcome to Bob, he can update his ensemble from $\left\{ \left(1/3,|\psi_{i}\rangle\right)\right\} ^{3}_{i=1}\text{ to }\left\{ \left(1/2,|\psi_{2}\rangle\right),\left(1/2,|\psi_{3}\rangle\right)\right\} $,
and then perform the optimal measurement for this reduced ensemble,
succeeding with probability $(1-s)$.

In the absence of communication, Bob instead performs the optimal
measurement for the original ensemble $\{(1/3,|\psi_{i}\rangle)\}^{3}_{i=1}$,
which consists of four measurement operators. However, the detection
operator corresponding to $|\psi_{1}\rangle$ never clicks, as its
support is orthogonal to $|\psi_{2}\rangle$ and $|\psi_{3}\rangle$.
Moreover, for any input state $|\psi_{i}\rangle$, the probability
of obtaining an inconclusive outcome is $s$. Since the states are
equiprobable, the overall failure probability remains $s$, regardless
of whether Bob receives Alice’s outcome. Thus, classical communication
does not improve the success probability. This argument extends straightforwardly
to the ensembles $\mathcal{E}_{N,k}$ for $k<N$. 

Our work raises several directions for future investigation. In Section \ref{sec:general}, we presented evidence suggesting a broader framework of which our construction appears to be the simplest instance. In particular, for the ensemble defined in Eq. \eqref{eq:gensemble}, our results indicate that $p_{L}^{me}\left(\mathcal{E}_{N,k}\right)<p_{G}^{me}\left(\mathcal{E}_{N,k}\right)$, while $p_{L}^{ud}\left(\mathcal{E}_{N,k}\right)=p_{G}^{ud}\left(\mathcal{E}_{N,k}\right)$ for all $N\geq3$ and $2\leq k\leq N$. We expect that the symmetry underlying the construction will play a crucial role in possibly proving this result.

A second open question concerns the precise separation between LOCC and global measurements in the minimum-error setting. For the examples studied here, we observe a nonzero, albeit small, gap between the optimal success probabilities achievable by PPT and global measurements, the latter being attained by the SRM. Whether the optimal LOCC performance coincides with that of PPT measurements, or exhibits a larger gap from the global optimum, remains an interesting open problem.

More broadly, our results raise a natural structural question: what properties of an ensemble give rise to a measure-dependent separation between global and LOCC discrimination \cite{comment1}? Our construction shares two notable features with the trine ensemble of Eq.~\eqref{ensemble-1}, which exhibits nonlocality without entanglement under both minimum-error and unambiguous discrimination. First, both ensembles are built from equiangular local states. Second, the local state labels are correlated: in Eq.~\eqref{ensemble-1}, all parties receive states carrying the same label, whereas in Eq.~\eqref{E(T)}, the labels are constrained to be distinct. Equivalently, the joint probability distribution of the local labels does not factorize in either case. Such correlations are necessarily present in any ensemble exhibiting a local–global separation, since independently chosen component states can be optimally discriminated by fixed local measurements in both the minimum-error and unambiguous settings \cite{gupta2024optimal}. Correlations among the local labels are therefore necessary, although clearly not sufficient, for nonlocality without entanglement.

The optimal unambiguous discrimination probability for $\mathcal{E}_{3,2}$, and the fact that it can be achieved by LOCC, also relies on the discrimination properties of the underlying equally separated states \cite{roa2011conclusive}. It would therefore be interesting to determine whether equiangularity and label correlations are fundamental ingredients of the present construction, useful sufficient conditions, or merely incidental features. In particular, can these be exploited to construct an ensemble exhibiting the opposite separation, \begin{equation}
    p_{L}^{me}=p_{G}^{me}\quad \text{and} \quad p_{L}^{ud}<p_{G}^{ud}, 
\end{equation}thereby reversing the behavior observed here? Conversely, can one do away with equiangularity while retaining suitable correlations among the local labels and still obtain a measure-dependent separation, perhaps with a larger gap between the optimal LOCC and global discrimination probabilities? 

\paragraph{Acknowledgment.}
TG acknowledges financial support from the ANRF National Post-Doctoral Fellowship (NPDF) under File No. PDF/2025/005147. We acknowledge the use of the computing resources at HPCE, IIT Madras.

\paragraph{Author contributions.}
All authors contributed equally.

\bibliographystyle{quantum}
\bibliography{bib}

\onecolumn
\appendix

\section*{Appendix}
In this Appendix we present the proof of Lemma~\ref{lem:main} and the Gram
matrix of the states of $\T$ and its square root. Due to the symmetry of the
problem in Lemma~\ref{lem:main}, two nonorthogonal ordered bases of $\C^3\otimes
\C^3$ lend themselves as natural bases to work in. These are (we suppress the $\otimes$ symbol to avoid clutter)

\begin{equation}\label{eq:bases}
     \begin{aligned}
         \B &= \{\ket{\psi_1}\ket{\psi_2},
         \ket{\psi_1}\ket{\psi_3},\ket{\psi_2}\ket{\psi_1},\ket{\psi_2}\ket{\psi_3},\ket{\psi_3}\ket{\psi_1},\ket{\psi_3}\ket{\psi_2},\{\ket{\psi_i}\ket{\psi_i}\}_{i=1}^3\}\\
         \B' &= \{\ket{\psi'_1}\ket{\psi'_2},
         \ket{\psi'_1}\ket{\psi'_3},\ket{\psi'_2}\ket{\psi'_1},\ket{\psi'_2}\ket{\psi'_3},\ket{\psi'_3}\ket{\psi'_1},\ket{\psi'_3}\ket{\psi'_2},\{\ket{\psi'_i}\ket{\psi'_i}\}_{i=1}^3\},
     \end{aligned}
\end{equation} 
where $\ket{\psi'_i}$ is orthogonal to all vectors in $\S$ except
$\ket{\psi_i}$, $i=1,2,3$. Since these are nonorthogonal bases, we present a few elementary
definitions from linear algebra in the following section, in order to review the
methods of constructing matrix representation of operators in a nonorthogonal
basis. All the material presented can be found in any standard linear algebra
textbook such as~\cite{axler2015linear}.

\section{Representing vectors and operators in nonorthogonal basis}
Let $V$ be a vector space of dimension $n$. Let $\B_1 = \{\ket{v_1}, \dots,
\ket{v_n}\}$ be an ordered basis of $V$. If $\ket{c} \in V$ has the form
\begin{equation}
    \ket{c} = \sum_{i=1}^n c_i \ket{v_i}
\end{equation} 
then the matrix representation of $\ket{c}$ \emph{with respect to the basis}
$\{\ket{v_1}, \dots, \ket{v_n}\}$ is given by 
\begin{equation}\label{eq:vec_matrix}
    \ket{c} = 
    \begin{pmatrix}
        c_1 \\ \vdots \\ c_n
    \end{pmatrix}.
\end{equation} 
It is conventional to write $\ket{c}$ for both the vector and its matrix
representation but we should bear in mind that the latter is basis dependent and
people usually suppress the basis when it is understood from the context. A
consequence of this definition is that, the basis vectors have a representation
consisting of only 1 and 0 in their own basis, 
\begin{equation}\label{eq:own_basis}
    \ket{v_i}=
    \begin{pmatrix}
        0 & \dots & 0 & \underbrace{1}_{\text{$i$-th position}} & 0 & \dots & 0 
    \end{pmatrix}^T
\end{equation} 
where $T$ denotes the transpose. Note that $\B_1$ need not be an orthogonal
basis. This implies that in the Hilbert space $\C^2$, instead of the canonical
basis of orthogonal states $\ket{0}$ and $\ket{1}$ if we use a basis consisting
of nonorthogonal states, say $\ket{0}$ and $\ket{+}=(\ket{0}+\ket{1})/\sqrt{2}$,
then these two will have the matrix representation  
\begin{equation}
    \ket{0}=
    \begin{pmatrix}
        1 \\ 0
    \end{pmatrix} 
    \qquad \text{and} \qquad 
    \ket{+}=
    \begin{pmatrix}
        0 \\ 1
    \end{pmatrix}
\end{equation} 
in their own basis, that is, $\{\ket{0},\ket{+}\}$. Of course, in a
nonorthogonal basis the inner product between two vectors 
\begin{equation}
    \ket{a}= 
    \begin{pmatrix}
        a_1 \\ \vdots \\ a_n
    \end{pmatrix} 
    \qquad \text{and} \qquad 
    \ket{b} = 
    \begin{pmatrix}
        b_1 \\ \vdots \\ b_n
    \end{pmatrix}
\end{equation} 
can no longer be computed by $\braket{a}{b}=\sum_i a_i^* b_i$, a formula which
holds only for orthogonal bases.

Let us now consider another vector space $W$ of dimension $m$ and consider a
linear transformation $A: V \rightarrow W$. In addition to the basis $\B_1$ in
$V$, fix another ordered basis $\B_2 = \{\ket{w_1}, \dots, \ket{w_n}\}$ in $W$.
Then the matrix of $A$ \textit{with respect to these bases} is the $m \times n$
matrix whose elements $a_{jk}$ are defined by 
\begin{equation}\label{eq:operator_matrix}
    A\ket{v_k}=a_{1k}\ket{w_1}+\cdots+a_{mk}\ket{w_m}.
\end{equation}
Like vectors, it is conventional to denote both the operator and its matrix
representation by $A$, but again, the latter is basis dependent which is usually
implicitly understood. A consequence of this definition is that we lose the
familiar matrix representation of operators in nonorthogonal bases. For example,
in the basis $\{\ket{0},\ket{+}\}$, the projection operator $\dyad{0}$ has the
representation 
\begin{equation}
    \dyad{0} = \begin{pmatrix}
        1 & \frac{1}{\sqrt{2}} \\ 0 & 0
    \end{pmatrix}.
\end{equation}

Let us now consider two different ordered bases $\B_1 = \{\ket{v_1}, \dots,
\ket{v_n}\}$ and $\B_2 = \{\ket{v'_1}, \dots, \ket{v'_n}\}$ of $V$, which are
related as 
\begin{equation}
    \ket{v_k} = \sum_{i=1}^n a_{ik} \ket{v'_i}.
\end{equation}
If $\ket{c} \in V$ can be written as
\begin{equation}
    \begin{aligned}
        \ket{c}&=c_1\ket{v_1} + \cdots + c_n \ket{v_n} \\
        &=c'_1\ket{v'_1} + \cdots + c'_n \ket{v'_n}
    \end{aligned}
\end{equation}
then the coefficients are related by
\begin{equation}
    \displaystyle
    \begin{pmatrix}
        c'_1 \\ \vdots \\ c'_n 
    \end{pmatrix} = \displaystyle 
    \begin{pmatrix}
        a_{11} & a_{12} & \dots & a_{1n} \\
        a_{21} & a_{22} & \dots & a_{2n} \\
        \vdots & \vdots & \ddots & \vdots \\
        a_{n1} & a_{n2} & \dots & a_{nn}
    \end{pmatrix} \displaystyle 
    \begin{pmatrix}
        c_1 \\ \vdots \\ c_n
    \end{pmatrix}.
\end{equation}

Let us finish this section with one last remark. Suppose we have a Hermitian
operator $A$ acting on an $n$-dimensional vector space $V$. We have its matrix
representation with respect to an arbitrary basis $\B_1$ and our task is to find
the matrix representation of some function of that operator $f(A)$. We can find
the matrix of $f(A)$ in the same basis $\B_1$ by the following procedure. First
we solve the eigenvalue problem; the eigenvectors will be obtained in the basis
$\B_1$. Let $\lambda_i$ and $\ket{e_i}$ denote the $i$-th eigenvalue and its
corresponding eigenvector respectively, and let the matrix of the $k$-th
eigenvector be 
\begin{equation}
    \ket{e_k}=\begin{pmatrix}
        e_{k1} \\ \vdots \\ e_{kn}
    \end{pmatrix} 
\end{equation} 
in the $\B_1$ basis. Construct the matrices $P$ and $D$ given by
\begin{equation}
    P = 
    \begin{pmatrix}
        \ket{e_1} & \ket{e_2} & \dots & \ket{e_n}
    \end{pmatrix}=
    \begin{pmatrix}
        e_{11} & e_{21} & \dots & e_{n1} \\ e_{12} & e_{22} & \dots & e_{n2} \\
        \vdots & \vdots & \ddots & \vdots \\
        e_{1n} & e_{2n} & \dots & e_{nn}
    \end{pmatrix}
\end{equation} 
and 
\begin{equation}
    D = \begin{pmatrix}
        \lambda_1 & 0 & \dots & 0 \\
        0 & \lambda_2 & \dots & 0 \\
        \vdots & \vdots & \ddots & 0 \\
        0 & 0 & \dots & \lambda_n
    \end{pmatrix}.
\end{equation} 
These matrices satisfy $AP=DP=PD$ from which it follows that 
\begin{equation}
    A = P DP^{-1}.
\end{equation} 
Note that the resultant matrix $P DP^{-1}$ of the operator $A$ is obtained in
the same basis $\B_1$ that we started in, and this basis need not be orthogonal.
To find the matrix of $f(A)$ in $\B_1$, we replace $D$ by $f(D)$ whose $i$-th
diagonal entry is $f(\lambda_i)$. Then we have $f(A)=Pf(D)P^{-1}$. Having
reviewed all the necessary definitions, we now move on to the proof of
Lemma~\ref{lem:main}.

\section{LOCC cannot distinguish the orthonormal basis $\tilde{\B}$} \label{main_lemma_proof}
Before going into the proof, we give a high-level view of our proof strategy for the readers' convenience. 

Since each element of $\{\ket{\mu_i}\}_{i=1}^6$ is entangled, there must be at least two product vectors appearing in the sum in \cref{eq:chen-sum}. As mentioned in the main text, our goal is to show that any two such vectors must be linearly dependent if they also satisfy the conditions of \cref{eq:chen-conditions}. We show this by exploiting a simple property of product states: for any bases $\{\ket{a_i}\}_i$ and $\{\ket{b_i}\}_i$ of $\mathcal{H}_A$ and $\mathcal{H}_B$ respectively, the coefficients of any product vector $\ket{\alpha}\ket{\beta}$ of $\mathcal{H}_A \otimes \mathcal{H}_B$ must satisfy some relations when it is expressed in the basis $\{\ket{a_i} \otimes \ket{b_j}\}_{i,j}$. For example, the product of the coefficients of $\ket{a_1b_1}$ and $\ket{a_2b_2}$ must equal that of $\ket{a_1b_2}$ and $\ket{a_2b_1}$. Note that no assumption of orthogonality is required for this to hold.

A natural basis of $\C^3$ to work in is given by the set $\S$ since it contains three linearly independent states. But upon closer inspection, the basis $\S^{\otimes 2}$ of $\C^3 \otimes \C^3$ (which is same as $\B$ of \cref{eq:bases}) presents a problem. Let $\ket{\alpha}\ket{\beta}$ be a product vector that appears in the decomposition of, say, $\ket{\mu_1}$. Since it is orthogonal to $\{\ket{\mu_i}\}_{i=2}^6$, we would like to express it in terms of $\ket{\mu_1}$ and a basis of $\W^\perp$. But none of the elements of $\B$ is orthogonal to $\W$, which makes the expression cumbersome if we were to use it to write $\ket{\alpha}\ket{\beta}$. To mitigate this problem we introduce the basis $\{\ket{\psi_i'}\}_{i=1}^3$ of $\C^3$ and use it to construct the basis $\B'$ containing the vectors $\{\ket{\psi_i'}\ket{\psi_i'}\}_{i=1}^3$ which span $\W^\perp$. This makes the expression of $\ket{\alpha}\ket{\beta}$ much simpler (see \cref{eq:b_prime} of proof below). We then work in this basis to derive the relations between the coefficients, as described in our proof strategy, yielding Equations \eqref{eq:ratio1} and \eqref{eq:ratio3}, from which the condition of linear dependency follows.

The last part of our proof is devoted to showing that certain coefficients of $\ket{\mu_1}$, which appear as denominators in Eq. \eqref{eq:ratio2}, does not equal 0 in the interval $s \in (0,1)$. This is done by first expressing $\ket{\mu_1}$ in $\B$ (\cref{eq:mu1_in_B}), then finding its expression in $\B'$ by a change of basis (\cref{eq:mu_1_in_B'}).

We now go into the technical details.

\begin{proof}[Proof of Lemma \ref{lem:main}]
    Suppose the states $\{\ket{\mu_i}\}$ are LOCC distinguishable. Then by
    Lemma~\ref{lem:chen}, the product vectors
    $\left\{\ket{\alpha_{ij}}\ket{\beta_{ij}}\right\} \subset \C^3 \otimes \C^3$
    that appear in the decomposition of $\ket{\mu_i}$ must belong to the
    4-dimensional subspace orthogonal to the 5 vectors
    $\left\{\ket{\mu_j}\right\}_{j=1, j\neq i}^6$. A nonorthogonal basis for
    this subspace is given by
    $\left\{\ket{\mu_i},\ket{\psi'_j}\ket{\psi'_j}\right\}_{j=1}^3$ where
    $\ket{\psi'_l}$ is orthogonal to all vectors in $\S$ except $\ket{\psi_l}$ and $l=1,2,3$.
    Therefore, any product vector that appears in the decomposition of
    $\ket{\mu_i}$, say $\ket{\alpha_{i1}}\ket{\beta_{i1}}$, can be written as 
    \begin{equation}\label{eq:chen}
        \ket{\alpha_{i1}}\ket{\beta_{i1}}=c_0\ket{\mu_i}+c_1\ket{\psi'_1}\ket{\psi'_1}+c_2\ket{\psi'_2}\ket{\psi'_2}+c_3\ket{\psi'_3}\ket{\psi'_3},
        \text{ where } c_0 \neq 0. 
    \end{equation} 
    
    The condition of $c_0 \neq 0$ is necessary since
    $\ket{\alpha_{i1}}\ket{\beta_{i1}}$ must be nonorthogonal to $\ket{\mu_i}$
    and $\left\{\ket{\psi'_j}\ket{\psi'_j}\right\}_{j=1}^3$ is orthogonal to the
    subspace $\W$ and hence to all $\ket{\mu_i}$'s. Since the $\ket{\mu_i}$'s
    are entangled, there must be at least two linearly independent product
    vectors $\ket{\alpha_{i1}}\ket{\beta_{i1}}$ and
    $\ket{\alpha_{i2}}\ket{\beta_{i2}}$ in the decomposition of $\ket{\mu_i}$
    for all $i$. In what follows, we show that this is impossible. 

    To have notational convenience, let us, without loss of generality, consider two
    product vectors $\ket{\alpha}\ket{\beta}$ and $\ket{\gamma}\ket{\delta}$ that
    appears in the decomposition of, say, $\ket{\mu_1}$ and write them as 
     \begin{equation}\label{eq:chen3}
         \begin{aligned}
             \ket{\alpha}\ket{\beta}&=a_0\ket{\mu_1}+a_1\ket{\psi'_1}\ket{\psi'_1}+a_2\ket{\psi'_2}\ket{\psi'_2}+a_3\ket{\psi'_3}\ket{\psi'_3},
             \; a_0 \neq 0, \\
            \ket{\gamma}\ket{\delta}&=b_0\ket{\mu_1}+b_1\ket{\psi'_1}\ket{\psi'_1}+b_2\ket{\psi'_2}\ket{\psi'_2}+b_3\ket{\psi'_3}\ket{\psi'_3},
            \; b_0 \neq 0.
         \end{aligned}
     \end{equation}
    Since $\ket{\alpha}\ket{\beta}$ is a product vector, we can expand each of
    $\ket{\alpha}$ and $\ket{\beta}$ in terms of the basis
    $\{\ket{\psi'_i}\}_{i=1}^3$ of $\C^3$, 
    \begin{equation}
        \begin{aligned}
            \ket{\alpha}&=\alpha_1\ket{\psi'_1}+\alpha_2\ket{\psi'_2}+\alpha_3\ket{\psi'_3} \\
            \ket{\beta}&=\beta_1\ket{\psi'_1}+\beta_2\ket{\psi'_2}+\beta_3\ket{\psi'_3}
        \end{aligned}
    \end{equation} 
    where $\alpha_i,\beta_i \in \C \; \forall i$. This gives us
    $\ket{\alpha}\ket{\beta}$ in the $\B'$ basis  
    \begin{multline}
        \ket{\alpha}\ket{\beta}=\alpha_1 \beta_2
        \ket{\psi'_1}\ket{\psi'_2}+\alpha_1 \beta_3
        \ket{\psi'_1}\ket{\psi'_3}+\alpha_2 \beta_1
        \ket{\psi'_2}\ket{\psi'_1}+\alpha_2 \beta_3 \ket{\psi'_2}\ket{\psi'_3}+\\
        \alpha_3 \beta_1 \ket{\psi'_3}\ket{\psi'_1}+\alpha_3 \beta_2
        \ket{\psi'_3}\ket{\psi'_2}+\alpha_1 \beta_1
        \ket{\psi'_1}\ket{\psi'_1}+\alpha_2 \beta_2
        \ket{\psi'_2}\ket{\psi'_2}+\alpha_3 \beta_3 \ket{\psi'_3}\ket{\psi'_3}.
    \end{multline} 
    Inspecting the above equation, we see that for any product vector $\ket{\zeta}
    \in \C^3 \otimes \C^3$ whose components are $\{\zeta_i\}_{i=1}^9$ in the
    ordering of the $\B'$ basis given by~\eqref{eq:bases}, the following conditions
    must be satisfied 

    \begin{equation}\label{eq:chen2}
        \begin{aligned}
             \zeta_1\zeta_9&=\zeta_2\zeta_6 \\
             \zeta_2\zeta_8&=\zeta_1\zeta_4 \\
             \zeta_4\zeta_7&=\zeta_2\zeta_3.
        \end{aligned}
    \end{equation} 
    Let the components of $\ket{\mu_1}$ with respect to $\B'$ be denoted by
    $\{\mu_{1i}\}_{i=1}^9$. We can use Equation~\eqref{eq:own_basis} to write
    $\ket{\psi'_i}\ket{\psi'_i}$ in the basis $\B'$,
    \begin{equation}
         \ket{\psi'_i}\ket{\psi'_i}=
         \begin{pmatrix}
             0 & \dots & 0 & \underbrace{1}_{\text{$(6+i)$-th position}} & 0 & \dots & 0 
         \end{pmatrix}^T.
    \end{equation} 
    This allows us to express $\ket{\alpha}\ket{\beta}$ in $\B'$ basis using
    Equation~\eqref{eq:chen3} as 
    \begin{equation}\label{eq:b_prime}
        \ket{\alpha}\ket{\beta}_{\B'}=
        \begin{pmatrix}
            a_0\mu_{11} & a_0\mu_{12} & a_0\mu_{13} & a_0\mu_{14} & a_0\mu_{15} & a_0\mu_{16} & a_0\mu_{17}+a_1 & a_0\mu_{18}+a_2 & a_0\mu_{19}+a_3
        \end{pmatrix}^T.
    \end{equation}  
Therefore, the conditions of Equation~\eqref{eq:chen2} for the state
$\ket{\alpha}\ket{\beta}$ read 
     \begin{equation}
         \begin{aligned}
             (a_0 \mu_{11})(a_0\mu_{19}+a_3)&=a_0\mu_{12}\, a_0 \mu_{16} \\
             (a_0 \mu_{12})(a_0\mu_{18}+a_2)&=a_0\mu_{11}\, a_0 \mu_{14} \\
             (a_0 \mu_{14})(a_0\mu_{17}+a_1)&=a_0\mu_{12}\, a_0 \mu_{13}
         \end{aligned}
     \end{equation} which can be written as \begin{equation} \label{eq:ratio1}
        a_i = k_i\, a_0, \; i=1,2,3
     \end{equation} where \begin{equation} \label{eq:ratio2}
         k_1 = \frac{\mu_{12}\mu_{13}}{\mu_{14}}-\mu_{17}, \quad k_2 = \frac{\mu_{11}\mu_{14}}{\mu_{12}}-\mu_{18} \; \text{ and }\; k_3 = \frac{\mu_{12}\mu_{16}}{\mu_{11}}-\mu_{19}.
     \end{equation} Similarly, the conditions of Equation~\eqref{eq:chen2} for $\ket{\gamma}\ket{\delta}$ read \begin{equation} \label{eq:ratio3}
        b_i = k_i\, b_0, \; i=1,2,3.
     \end{equation} 
     We will soon show that for $s \in (0,1)$ the components $\mu_{11},
     \mu_{12}$ and $\mu_{14}$ do not equal 0, so that Equation~\eqref{eq:ratio2}
     and hence Equations~\eqref{eq:ratio1} and \eqref{eq:ratio3} are legitimate.
     For now let us consider Equations~\eqref{eq:ratio1} and \eqref{eq:ratio3}.
     These equations imply that the vectors $\ket{\alpha}\ket{\beta}$ and
     $\ket{\gamma}\ket{\delta}$ can be written as 
     \begin{equation}
         \begin{aligned}
             \ket{\alpha}\ket{\beta}&=a_0\left(\ket{\mu_1}+k_1\ket{\psi'_1}\ket{\psi'_1}+k_2\ket{\psi'_2}\ket{\psi'_2}+k_3\ket{\psi'_3}\ket{\psi'_3}\right), \; a_0 \neq 0, \\
     \ket{\gamma}\ket{\delta}&=b_0\left(\ket{\mu_1}+k_1\ket{\psi'_1}\ket{\psi'_1}+k_2\ket{\psi'_2}\ket{\psi'_2}+k_3\ket{\psi'_3}\ket{\psi'_3}\right), \; b_0 \neq 0
         \end{aligned}
     \end{equation} which shows that they are linearly dependent. Therefore, any two product vectors in the span of $\{\ket{\mu_i},\{\ket{\psi'_j}\ket{\psi'_j}\}_{j=1}^3\}$ with nonzero overlap with $\ket{\mu_i}$ are linearly dependent. This implies that the necessary conditions of Lemma \ref{lem:chen} are not met, and as a result the states $\{\ket{\mu_i}\}$ cannot be distinguished by LOCC.

     We now show that we are not dividing by zero in Equation~\eqref{eq:ratio2} in
     the interval $s \in (0,1)$, by equating $\mu_{11}, \mu_{12}$ and $\mu_{14}$
     to zero and solving for $s$. To this end, we first write the components of
     $\ket{\mu_1}$ in the $\B$ basis and then transform them to $\B'$. The
     components of $\ket{\mu_1}$ w.r.t. $\B$ can be obtained by noting that
     $\ket{\mu_1}=\rho^{-1/2}\ket{\psi_1}\ket{\psi_2}$, and consequently
     $\ket{\mu_1}$ is given by the first column of the matrix of $\rho^{-1/2}$
     in $\B$. If $\{\lambda_i\}$ and $\{\ket{\lambda_i}\}$ are the eigenvalues
     and eigenvectors of $\rho$ then $\rho^{-1/2}$ can be computed as
     
     \begin{equation}\label{eq:sigma_inv}
        \rho^{-1/2}=\sum_i \frac{1}{\sqrt{\lambda_i}}\dyad{\lambda_i}, \;
        \lambda_i \neq 0,
    \end{equation} 
    and the eigenvectors corresponding to the zero eigenvalues do not contribute to
    the sum in Equation~\eqref{eq:sigma_inv}. To write down $\rho$ in $\B$, we note that 
    \begin{equation}
        \rho=\frac{1}{6}\sum_{i=1}^6 \dyad{\phi_i}.
    \end{equation} 
    Therefore the matrix of $\rho$ in $\B$ can be found by its action on the basis
    vectors as given by Equation~\eqref{eq:operator_matrix}
    \begin{equation}
        \begin{aligned}
            \rho \ket{\phi_j}&=\left(\frac{1}{6}\sum_{i=1}^6 \dyad{\phi_i}\right)\ket{\phi_j}\\
            &=\frac{1}{6}\sum_{i=1}^6 \braket{\phi_i}{\phi_j}\ket{\phi_i}
        \end{aligned}
    \end{equation} 
    where $j \in [9]$ and $\{\ket{\phi_i}\}_{i=7}^9$ are given by
    $\{\ket{\psi_i}\ket{\psi_i}\}_{i=1}^3$, respectively. This gives us 
    \begin{equation}
        \rho_\B=\frac{1}{6}\left(
        \begin{array}{ccccccccc}
             1 & s & s^2 & s^2 & s^2 & s & s & s & s^2 \\
             s & 1 & s^2 & s & s^2 & s^2 & s & s^2 & s \\
             s^2 & s^2 & 1 & s & s & s^2 & s & s & s^2 \\
             s^2 & s & s & 1 & s^2 & s^2 & s^2 & s & s \\
             s^2 & s^2 & s & s^2 & 1 & s & s & s^2 & s \\
             s & s^2 & s^2 & s^2 & s & 1 & s^2 & s & s \\
             0 & 0 & 0 & 0 & 0 & 0 & 0 & 0 & 0 \\
             0 & 0 & 0 & 0 & 0 & 0 & 0 & 0 & 0 \\
             0 & 0 & 0 & 0 & 0 & 0 & 0 & 0 & 0 \\
        \end{array}
        \right).
    \end{equation}
    From this we can compute $\rho^{-1/2}$ in $\B$ as
    \begin{equation}
        \rho^{-1/2}_\B = \frac{1}{\sqrt{6}}    
        \begin{pmatrix}
            \gamma_0  & \gamma_1  & \gamma_2  & \gamma_3  & \gamma_3  & \gamma_1  & \gamma_4  & \gamma_4  & \gamma_5  \\
            \gamma_1  & \gamma_0  & \gamma_3  & \gamma_1  & \gamma_2  & \gamma_3  & \gamma_4  & \gamma_5  & \gamma_4  \\
            \gamma_2  & \gamma_3  & \gamma_0  & \gamma_1  & \gamma_3  & \gamma_1  & \gamma_4  & \gamma_4  & \gamma_5  \\
            \gamma_3  & \gamma_1  & \gamma_1  & \gamma_0  & \gamma_3  & \gamma_2  & \gamma_5  & \gamma_4  & \gamma_4  \\
            \gamma_3  & \gamma_2  & \gamma_3  & \gamma_3  & \gamma_0  & \gamma_1  & \gamma_4  & \gamma_5  & \gamma_4  \\
            \gamma_1  & \gamma_3  & \gamma_1  & \gamma_2  & \gamma_1  & \gamma_0  & \gamma_5  & \gamma_4  & \gamma_4  \\
            0  & 0  & 0  & 0  & 0  & 0  & 0  & 0  & 0 \\
            0  & 0  & 0  & 0  & 0  & 0  & 0  & 0  & 0 \\
            0  & 0  & 0  & 0  & 0  & 0  & 0  & 0  & 0 \\
        \end{pmatrix}
    \end{equation}
    where 
    \begin{align}
        \gamma_0 = 2v_0 + 2v_1 + v_2 + v_3, \quad 
        \gamma_1 = -v_0 + v_1 - v_2 + v_3, \quad
        \gamma_2 = 2v_0 - 2v_1 - v_2 + v_3, \notag \\
        \gamma_3 = -v_0 - v_1 + v_2 + v_3, \quad
        \gamma_4 = 2v_4 + 2v_5, \quad
        \gamma_5 = -4v_4 + 2v_5
    \end{align}
    and 
    \begin{align}
        v_0 = \frac{1}{\sqrt{(1 - s)}}, \quad 
        v_1 = \frac{1}{\sqrt{-2s^2 + s + 1}}, \quad
        v_2 = \frac{1}{\sqrt{s^2 - 2s + 1}}, \notag \\
        v_3 = \frac{1}{\sqrt{3s^2 + 2s + 1}}, \quad
        v_4 = \frac{s}{\sqrt{1 - s}}, \quad
        v_5 = \frac{s^2 + 2s}{\sqrt{3s^2 + 2s + 1}(9s^2 + 6s + 3)}.
    \end{align}
    Reading off the first column of $\rho^{-1/2}$, we have 
    \begin{equation}\label{eq:mu1_in_B}
        \ket{\mu_1}_\B=\frac{1}{\sqrt{6}} \left(
        \begin{array}{c}
             \frac{1}{\sqrt{3 s^2+2 s+1}}+\frac{2}{\sqrt{-2 s^2+s+1}}+\frac{2}{\sqrt{1-s}}+\frac{1}{1-s} \\
             \frac{1}{\sqrt{3 s^2+2 s+1}}+\frac{1}{\sqrt{-2 s^2+s+1}}-\frac{1}{\sqrt{1-s}}+\frac{1}{s-1} \\
             \frac{1}{\sqrt{3 s^2+2 s+1}}-\frac{2}{\sqrt{-2 s^2+s+1}}+\frac{2}{\sqrt{1-s}}+\frac{1}{s-1} \\
             \frac{1}{\sqrt{3 s^2+2 s+1}}-\frac{1}{\sqrt{-2 s^2+s+1}}-\frac{1}{\sqrt{1-s}}+\frac{1}{1-s} \\
             \frac{1}{\sqrt{3 s^2+2 s+1}}-\frac{1}{\sqrt{-2 s^2+s+1}}-\frac{1}{\sqrt{1-s}}+\frac{1}{1-s} \\
             \frac{1}{\sqrt{3 s^2+2 s+1}}+\frac{1}{\sqrt{-2 s^2+s+1}}-\frac{1}{\sqrt{1-s}}+\frac{1}{s-1} \\
             0 \\
             0 \\
             0 \\
        \end{array}
        \right)
    \end{equation} 
    in the $\B$ basis.
    
    To find $\ket{\mu_1}$ in $\B'$, we need to multiply the matrix of Equation~\eqref{eq:mu1_in_B} by the change-of-basis matrix $M$ whose columns are the
    coefficients of the elements of $\B$ in $\B'$. That is, the $(i,j)$-th element
    of $M$ will be the $i$-th coefficient of $\ket{\phi_j}$ in terms of $\B'$. This
    can be found out by first expressing the states of $\{\ket{\psi_i}\}$ in terms
    of $\{\ket{\psi'_i}\}$. By straightforward calculation we can find the
    relationship between $\B$ and $\B'$, 
    \begin{equation}
        \begin{aligned}
            \ket{\psi'_i} & = \sqrt{\frac{1+s}{1+s-2s^2}}\left(\ket{\psi_i} -
            \frac{s}{1+s} \sum_{\substack{j=1 \\ j \neq i}}^3
            \ket{\psi_j}\right), \\ \ket{\psi_i} & =
            \sqrt{\frac{1+s}{1+s-2s^2}}\left(\ket{\psi'_i} + s
            \sum_{\substack{j=1 \\ j \neq i}}^3 \ket{\psi'_j}\right), \quad
            i=1,2,3.
        \end{aligned}
    \end{equation} 
    We now express the vectors of $\B$ in terms of the basis $\B'$ and construct the
    matrix $M$ 
    \begin{equation}
        M = \frac{1+s}{1+s-2s^2}\,
        \begin{pmatrix}
            1   & s   & s^2 & s^2 & s^2 & s   & s   & s   & s^2 \\
            s   & 1   & s^2 & s   & s^2 & s^2 & s   & s^2 & s   \\
            s^2 & s^2 & 1   & s   & s   & s^2 & s   & s   & s^2 \\
            s^2 & s   & s   & 1   & s^2 & s^2 & s^2 & s   & s   \\
            s^2 & s^2 & s   & s^2 & 1   & s   & s   & s^2 & s   \\
            s   & s^2 & s^2 & s^2 & s   & 1   & s^2 & s   & s   \\
            s   & s   & s   & s^2 & s   & s^2 & 1   & s^2 & s^2 \\
            s   & s^2 & s   & s   & s^2 & s   & s^2 & 1   & s^2 \\
            s^2 & s   & s^2 & s   & s   & s   & s^2 & s^2 & 1  
        \end{pmatrix}.
    \end{equation}
    We can now get $\ket{\mu_1}$ in $\B'$ by multiplying its matrix in $\B$ by $M$.
    This gives us 
    \begin{equation}\label{eq:mu_1_in_B'}
        \ket{\mu_1}_{\B'}= \sqrt{6}
        \begin{pmatrix}
            -\frac{(s+1) \left(2 \sqrt{-2 s^2+s+1}-s+2 \sqrt{1-s}+\sqrt{s (3
            s+2)+1}+1\right)}{6 (s-1) (2 s+1)}\\ \frac{(s+1) \left(-\sqrt{-2
            s^2+s+1}-s+\sqrt{1-s}-\sqrt{s (3 s+2)+1}+1\right)}{6 (s-1) (2
            s+1)}\\ -\frac{(s+1) \left(-2 \sqrt{-2 s^2+s+1}+s+2
            \sqrt{1-s}+\sqrt{s (3 s+2)+1}-1\right)}{6 (s-1) (2 s+1)}\\
            \frac{(s+1) \left(\sqrt{-2 s^2+s+1}+s+\sqrt{1-s}-\sqrt{s (3
            s+2)+1}-1\right)}{6 (s-1) (2 s+1)}\\ \frac{(s+1) \left(\sqrt{-2
            s^2+s+1}+s+\sqrt{1-s}-\sqrt{s (3 s+2)+1}-1\right)}{6 (s-1) (2
            s+1)}\\ \frac{(s+1) \left(-\sqrt{-2 s^2+s+1}-s+\sqrt{1-s}-\sqrt{s (3
            s+2)+1}+1\right)}{6 (s-1) (2 s+1)}\\ \frac{s (s+1) \left(2
            \sqrt{1-s}+s \left(\sqrt{1-s}-\sqrt{s (3 s+2)+1}\right)+\sqrt{s (3
            s+2)+1}\right)}{3 (1-s)^{3/2} (2 s+1) \sqrt{s (3 s+2)+1}}\\ \frac{s
            (s+1) \left(2 \sqrt{1-s}+s \left(\sqrt{1-s}-\sqrt{s (3
            s+2)+1}\right)+\sqrt{s (3 s+2)+1}\right)}{3 (1-s)^{3/2} (2 s+1)
            \sqrt{s (3 s+2)+1}}\\ \frac{s (s+1) \left(s \sqrt{1-s}+2
            \sqrt{1-s}+2 s \sqrt{s (3 s+2)+1}-2 \sqrt{s (3 s+2)+1}\right)}{3
            (1-s)^{3/2} (2 s+1) \sqrt{s (3 s+2)+1}} 
        \end{pmatrix}
    \end{equation} 
    in the basis $\B'$. To find the values of $s$ for which $\mu_{11}, \mu_{12}$
    and $\mu_{14}$ equal zero, we set 
    \begin{subequations}
        \begin{equation} \label{eq:a}
            \mu_{11}=-\frac{(s+1) \left(2 \sqrt{-2 s^2+s+1}-s+2 \sqrt{1-s}+\sqrt{s (3 s+2)+1}+1\right)}{\sqrt{6} (s-1) (2 s+1)}=0
        \end{equation}
        \begin{equation}\label{eq:b}
            \mu_{12}=\frac{(s+1) \left(-\sqrt{-2 s^2+s+1}-s+\sqrt{1-s}-\sqrt{s (3 s+2)+1}+1\right)}{\sqrt{6} (s-1) (2 s+1)}=0
        \end{equation}
        \begin{equation}\label{eq:c}
           \mu_{14}= \frac{(s+1) \left(\sqrt{-2 s^2+s+1}+s+\sqrt{1-s}-\sqrt{s (3 s+2)+1}-1\right)}{\sqrt{6} (s-1) (2 s+1)} =0.
        \end{equation}
    \end{subequations}
    Solving these three equations, we get as solution $s=-1$ for Equation~\eqref{eq:a} and $s=-1,0$ for Equations~\eqref{eq:b} and \eqref{eq:c}. This completes our proof.
\end{proof}

\section{Gram matrix and its square root}\label{gram_matrix}
The Gram matrix $\Gamma$ of the set of states $\T$ is given by 
\begin{equation} 
    \label{gram}
    \Gamma = \left(
    \begin{array}{cccccc}
         1 & s & s^2 & s^2 & s^2 & s \\
         s & 1 & s^2 & s & s^2 & s^2 \\
         s^2 & s^2 & 1 & s & s & s^2 \\
         s^2 & s & s & 1 & s^2 & s^2 \\
         s^2 & s^2 & s & s^2 & 1 & s \\
         s & s^2 & s^2 & s^2 & s & 1 \\
    \end{array}
    \right).
\end{equation}
Its square root is 
\begin{equation}
    \label{gram_sqrt}
    \sqrt{\Gamma} = \frac{1}{6} 
    \begin{pmatrix}
        \gamma_0  & \gamma_1  & \gamma_2  & \gamma_3  & \gamma_3  & \gamma_1  \\
        \gamma_1  & \gamma_0  & \gamma_3  & \gamma_1  & \gamma_2  & \gamma_3  \\
        \gamma_2  & \gamma_3  & \gamma_0  & \gamma_1  & \gamma_1  & \gamma_3  \\
        \gamma_3  & \gamma_1  & \gamma_1  & \gamma_0  & \gamma_3  & \gamma_2  \\
        \gamma_3  & \gamma_2  & \gamma_1  & \gamma_3  & \gamma_0  & \gamma_1  \\
        \gamma_1  & \gamma_3  & \gamma_3  & \gamma_2  & \gamma_1  & \gamma_0  \\
    \end{pmatrix}
\end{equation}
where
\begin{align}
    \gamma_0 = 2 v_0 + 2 v_1 + v_2 + v_3, \quad 
    \gamma_1 = v_1 - v_0 - v_2 + v_3, \notag \\
    \gamma_2 = -2 v_1 + 2 v_0 - v_2 + v_3, \quad
    \gamma_3 = - v_1 - v_0 + v_2 + v_3
\end{align}
and
\begin{equation}
    v_0 = \sqrt{1 - s}, \quad 
    v_1 = \sqrt{-2s^2 + s + 1}, \quad
    v_2 = \sqrt{(s - 1)^2}, \quad
    v_3 = \sqrt{3s^2 + 2s + 1}.
\end{equation} Note that all the diagonal entries of $\sqrt{\Gamma}$ are equal.

\end{document}